\documentclass[a4paper,11pt]{article}
\usepackage{graphicx}
\usepackage{here}
\usepackage{amsmath,amssymb,amsthm}
\usepackage{mathrsfs}
\usepackage[sort, comma]{natbib}
\usepackage[subrefformat=parens]{subcaption}
\usepackage{url}
\usepackage{tikz}
\usetikzlibrary{arrows.meta,intersections,calc,decorations.pathreplacing,math}
\usepackage{pgfplots}
\pgfplotsset{compat=1.18}
\usepackage{mathtools}
\usepackage{physics}
\mathtoolsset{showonlyrefs=true}
\usepackage[stable]{footmisc}
\usepackage{geometry}
\usepackage{caption}
\newtheorem{prop}{Proposition}
\newtheorem{lem}{Lemma}

\newcommand{\ol}[1]{\overline{#1}}

\let\Re\relax
\DeclareMathOperator{\Re}{Re}

\begin{document}
\begin{titlepage}
\begin{center}
{\LARGE City formation by dual migration of firms and workers} \\
\vspace{10mm}
{\Large Kensuke Ohtake}\\
\vspace{10mm}
{\Large Center for General Education, Shinshu University, Matsumoto, Nagano 390-8621, Japan. \\
\vspace{2mm}
{Corresponding author's e-mail address:\\
\vspace{1mm}
{\underline{\bf k\_ohtake@shinshu-u.ac.jp}}}}\\

\vspace{10mm}
{\Large July 30, 2026}
\end{center}
\end{titlepage}

\begin{abstract}
The Core-Periphery model in the new economic geography, which considers the single migration of workers driven by real wage inequality among regions, is extended to incorporate the migration of firms driven by real profit inequality among regions. In this dual-migration model, the behavior of solutions is qualitatively similar to that of single-migration models. That is, 1) spatially homogeneous population distributions become destabilized and eventually form several cities where both firms and workers agglomerate; 2) the number of cities decreases as transport costs decrease. These results provide a more general theoretical justification for the use of single migration models.
\end{abstract} 

\noindent
{\bf Keywords:\hspace{1mm}}
city formation;
core-periphery model;
economic agglomeration;
new economic geography;
self-organization;
transport costs


{\section{Introduction}}
The Core-Periphery (CP) model proposed by \citet{Krug91} is one of the most representative mathematical models based on microeconomic theory to describe economic agglomeration, such as city formation.\footnote{For example, \citet{FujiKrugVenab} and \citet{FujiThis} give an overview of developments in this field.} Many mathematical models descended from the CP model are based on the migration of workers driven by real wage inequality among regions; workers flow out of regions with lower-than-average real wages and into regions with higher-than-average real wages. In this case, the migration of firms is modeled as instantaneously adjusting to worker migration. In such ``single migration'' models, excess profits disappear due to competition arising from the free entry of firms.

In reality, however, excess profits remain more or less, and firms migrating in search of excess profits in the spatial economy cannot be ignored. In this context, we must mention pioneering models devised by \citet{Puga99}\footnote{In the Puga model, profits in the region where firms exist are zero in the long-run equilibrium. On the other hand, in the models by \citet{PicaThiToul04}, excess profits do not disappear because free entry is restricted.} and \citet{PicaThiToul04},\footnote{They have derived various models depending on players to which profits are distributed, and the present paper is based on what they call a ``entrepreneurs model'' in which all profits belong to owners of firms, who are not workers.} which incorporates the profit-driven migration of firms. In their models, firms migrate in a gradual response to profit inequality between two regions, while workers adjust instantaneously so that real income equality between the two regions is maintained.\footnote{\citet[p.~310, Footnote 11 on p.~311]{Puga99} and \citet[p.~155]{PicaThiToul04}} Hence, we may regard their models as other single migration models. Under such firm-led single migration models, results similar to the standard  NEG models, such as the CP model, are obtained, and \citet[p.~298]{FujiThis}, citing \citet{muth1971migration}, conclude that ``{\it it seems fair to say that NEG is agnostic about ``who'' is the egg and the chicken}''. However, we would like to pose the question: if both firms and workers could migrate independently, rather than one being subordinate to the other, what would the conclusion of NEG be?

This paper generalizes the single-migration model and proposes a dual-migration model in which both firms and workers gradually migrate, driven by their own interests. To this end, we extend the entrepreneurs model proposed by \citet{PicaThiToul04}, which is a firm-led, single-migration model, to one in a continuous space, while modifying it so that worker migration also progresses gradually in response to real wage inequality. Then, we obtain a model with ``dual migration process'' of firms and workers. The behavior of solutions of the model is then investigated in detail. In particular, we consider the model on a continuous one-dimensional circumference and focus on agglomeration patterns formed by the dual migration process.

The results of this paper are summarized as follows. First, we analyze the stability of a homogeneous stationary solution where all quantities ---namely, population of firms, population of workers, nominal wages, price indices, and real wages--- are uniformly distributed along the circumference. By the Fourier method, a small perturbation added to the homogeneous stationary solution can be expressed as a superposition of countably infinite eigenmodes, each having its own unique spatial frequency. Each mode decays over time under sufficiently high transport costs, while it grows over time under sufficiently low transport costs. Therefore, there exists a critical point of transport costs that serves as the threshold separating temporal decay from growth of each eigenmode. The value of this critical point decreases as the absolute value of the spatial frequency of each eigenmode becomes smaller. Second, we analytically show that, despite firms and workers migrating independently, the agglomerated regions of these two players coincide in any stationary solutions. Finally, we perform numerical simulations of the time evolution of solutions starting near the homogeneous stationary solution to clarify the shape of non-homogeneous stationary solutions that asymptotically form. These non-homogeneous solutions, like in standard NEG models in continuous space, exhibit spiky population distributions, and the number of spikes decreases as transport costs fall.

Let us list several other literature related to this research. The results of this paper, in which a homogeneous stationary solution is destabilized and forms spatial structures, can be placed in the context of symmetry-breaking in the spatial economy or, more broadly, in the science of self-organization. See \citet{Krug} or \citet{Matsu2008sym} for example. In the context of NEG, although not as numerous as discrete regional models, several studies have also examined continuous space models. The results of linear stability analysis using Fourier analysis obtained by \citet[Chapter 6]{FujiKrugVenab} and \citet{OhtakeYagi_Asym} in the single migration model are also valid for the dual migration model. \citet{TabaEshi_explosion} and \citet{Fab} provide analytical discussions regarding the spike-like stationary solutions of some NEG models. \citet{OhtakeYagi_Asym}, \citet{Ohtake2023cont}, and \citet{Ohtake2025unique} numerically computed the asymptotic behavior of solutions to the standard NEG models on the circumference and found that spiky distributions emerge, and the number of spikes decreasing as transport costs fall.\footnote{Other literature dealing with economic geography models in continuous space includes \citet{Moss}, \citet{MossMar2004}, \citet{ChiAsh08}, \citet{PicaTabu}, \cite{TabaEshiSakaTaka}, \citet{TabaEshi_existence}, \citet{GoKoTa2017cont}, and \citet{TabaEshi23}.} Although the models addressed in these studies are single migration models, the asymptotic behavior of their solutions closely resembles that of the dual migration model presented in this paper. The consistency of these results regarding solution behavior suggests that the more general dual-migration model can justify the theoretical simplification of the single-migration model. In the context of introducing different types of mobile populations, \citet{dePaPapaThis2019origin} construct an abstract model of spatial economy and investigate its non-homogeneous stationary solution. In their model, when two different types of populations are introduced, it is shown that there are spatial patterns in which they do not agglomerate in the same area. This contrasts with the model in this paper, where the different players (mobile firms and workers) eventually agglomerate in the same regions. This difference stems from the fact that their model allows for the two types of populations to interact in a mutually avoiding manner.

The rest of the paper is organized as follows. Section 2 presents the model treated in this paper. Section 3 investigates the stability of a homogeneous stationary solution. Section 4 shows analytical results on stationary solutions and numerical results on the asymptotic behavior of solutions of the model. Section 5 concludes. Section 6 is Appendix that provides some content omitted from the main text.

{\vspace{5mm}}
\section{The model}

\subsection{Settings and assumptions}
An infinite number of regions are continuously and uniformly spaced on a one-dimensional circle with radius $\rho>0$ denoted by $S$. Each region is denoted by $x$, $y$, or $z\in S$. 

The economy consists of two sectors: manufacturing and agriculture. In manufacturing, an infinite number of varieties of goods are produced by firms under monopolistic competition. One variety of manufactured goods is produced by one firm. Same as \citet{PicaThiToul04}, we do not assume increasing returns to scale in manufacturing. Instead, a constant-return technology with no fixed costs is assumed. In the agricultural sector, a single homogeneous variety of agricultural good is produced under perfect competition with a constant-return technology. The transportation of manufactured goods among regions incurs so-called iceberg transport costs,\footnote{This is introduced by \citet{Sam1952}. That is, a fraction of a good vanishes during transportation. Hence, one or more units of a good must be shipped to deliver one unit of the good to its destination.} while an agricultural good does not.

There are two types of workers: manufacturing and agricultural. Manufacturing workers are employed by firms and engaged in producing a large variety of manufactured goods. Agricultural workers produce a homogeneous agricultural good. The workers cannot change the sectors in which they are employed. The manufacturing workers can move freely between regions, while agricultural workers cannot. The total population of each of these two types of workers is finite and constant.

Entrepreneurs own firms, each owning one firm. Like manufacturing workers, entrepreneurs can move around in space. The total population of entrepreneurs is also finite and constant, and each firm is indivisible. As a result, since the total number of firms is also fixed, this leads to imperfect competition. Therefore, as \citet[p.~147]{PicaThiToul04} state, it is not explicit increasing returns to scale that generate agglomeration forces in this model, but rather the combination of imperfect competition and indivisibility of  firms.

All consumers, consisting of manufacturing workers, agricultural workers, and entrepreneurs, are assumed to have identical utility functions. The only source of the workers' income is nominal wages, and the only source of the entrepreneurs' income is excess profits. In other words, the entrepreneurs acquire all excess profits generated in the manufacturing sector. All consumers use their income only for consumption without saving or investing. Therefore, their welfare level is measured by their real income, which is their nominal income deflated by the price indices of goods. Firms (Entrepreneurs) and manufacturing workers gradually migrate from regions where their real incomes are below average to regions where they are above average.

\subsection{Consumer behavior}
As in \citet{PicaThiToul04}, we use the Cobb-Douglas utility function
\begin{equation}\label{utility}
U=C^\mu A^{1-\mu},\hspace{3mm}\mu\in[0,1).
\end{equation}
of each consumer. Here, $C$ and $A$ are the consumption of a composite index of manufactured goods and the consumption of the agricultural good, respectively. There are an infinite number of varieties of manufactured goods, and the range of the variety is denoted by a closed interval $[0,n]\subset\mathbb{R}$. Let $c(h)$ be the consumption of the $h$-th variety of the manufactured goods, where $h\in[0,n]$. The composite index of the manufactured goods is defined by
\[
C = \left[\int_0^n c(h)^{\frac{\sigma-1}{\sigma}}  dh\right]^{\frac{\sigma}{\sigma-1}},
\]
where $\sigma>1$ stands for the elasticity of substitution between any two varieties of manufactured goods. The price of the $h$-th variety of manufactured goods is denoted by $p(h)$ while the agricultural good is taken as the num\'{e}raire, and its price is fixed to $1$. The budget constraint is then
\begin{equation}\label{budget}
\int_0^n p(h)c(h) dh + A = Y,
\end{equation}
where $Y$ is the nominal income of the consumer.

Maximizing \eqref{utility} under \eqref{budget} yields the following demand function for each variety
\begin{equation}\label{demand}
c(h)= \mu Yp(h)^{-\sigma}G^{\sigma-1}.
\end{equation}
Here, $G$ is a price index of manufactured goods defined by
\begin{equation}\label{priceindex}
G = \left[\int_0^n p(h)^{1-\sigma}dh\right]^{\frac{1}{1-\sigma}}.
\end{equation}
It is known\footnote{\citet[p.47]{FujiKrugVenab}} that
\[
\int_0^n p(h)c(h)dh = GC
\]
and thus the budget constraint \eqref{budget} is now 
\begin{equation}\label{budget_composite}
GC + A = Y,
\end{equation}
which means that the price index $G$ can be regarded as the price of one unit of the composite index of manufactured goods.

\subsection{Continuous space modeling}\label{subsec:contspace}
We consider a continuous-space model, but the fundamental derivation approach is similar to that for discrete multiple-region models, as in \citet[Chapter 4]{FujiKrugVenab}. To simplify notation, the time variable $t$ is generally not explicitly stated even for time-dependent functions in Subsections \ref{subsec:contspace}-\ref{subsec:modeleqs}.

In this section, the stage for economic activity is a subset $\Omega$ of Euclidean space $\mathbb{R}^n$ where $n\geq 1$. Let $\Phi$, $N$, and $M$ be the total population of agricultural workers, the total population of firms, and the total population of manufacturing workers, respectively. Let $\phi(x)$, $n(x)$, and $m(x)$ be the share of region $x\in\Omega$ in the total population of agricultural workers, firms, and manufacturing workers, respectively. The population density of agricultural workers (resp. firms, manufacturing workers) in region $x$ is then given by $\Phi\phi(x)$ (resp. $Nn(x)$, $Mm(x)$). Note that $\phi(x)\geq 0$, $n(x)\geq 0$, and $m(x)\geq 0$ in any region $x\in\Omega$, and that
\begin{equation}\label{conservation}
\int_\Omega\phi(x)dx = \int_\Omega n(x)dx = \int_\Omega m(x)dx\equiv 1.
\end{equation}
Here, the time variables $t$ in $n$ and $m$ are omitted, but the conservation law \eqref{conservation} is assumed to hold for any time $t\geq 0$.

The iceberg transportation is assumed, that is, $T(x, y)(=T(y, x))\geq 1$ units of a variety of manufactured goods must be shipped from region $x$ to deliver one unit of the variety to region $y$. In addition, it is assumed that all varieties of manufactured goods produced in region $x$ have the same price, denoted $p(x)$, in that region. Therefore, the price in region $y$ of any variety of manufactured goods produced in region $x$ is given by $p(x, y)=p(x)T(x, y)$. 

The price index \eqref{priceindex} then takes different values for different regions as
\begin{align}
G(x) &= \left[\int_\Omega Nn(y)p(y, x)^{1-\sigma}dy\right]^{\frac{1}{1-\sigma}}\\
&= \left[\int_\Omega Nn(y) \left(p(y)T(x, y)\right)^{1-\sigma}dy\right]^{\frac{1}{1-\sigma}}.\label{priceindex_r}
\end{align}

By the demand function \eqref{demand}, the total demand from region $y$ for one variety of manufactured goods produced in region $x$ is
\begin{equation}\label{cs}
c(y) = \mu Y(y) p(x,y)^{-\sigma}G(y)^{\sigma-1}.
\end{equation}
Here, $Y(y)$ is the total nominal income of region $y$, which is the sum of the incomes of all consumers in region $y$, i.e., manufacturing workers, agricultural workers, and entrepreneurs in region $y$.\footnote{To be precise, since the density rather than the number of people in each region is given, $Y(y)$ is not exactly the total income itself, but rather represents its spatial density. However, to avoid redundancy, the term ``density'' may be omitted when referring to a quantity that should strictly be called density, unless particularly necessary.} Therefore, the total income in region $x$ is given by
\begin{equation}\label{totalincome}
Y(x) =  Mw(x)m(x) + \Phi\phi(x) + Ng(x)n(x),
\end{equation}
where $g(x)$ and $w(x)$ are the nominal profit of each firm in region $x$, and the nominal wage of each manufacturing worker in region $x$, respectively. It should be noted that the nominal wage of each agricultural worker is fixed to $1$ in any region.\footnote{As in \citet{PicaThiToul04}, the agricultural good is assumed to be produced under perfect competition with a constant-return technology (Specifically, one unit of agricultural workers produces one unit of the agricultural good.). In this case, the nominal wage of agricultural workers is equal to the nominal price of the agricultural good ($=1$) in market equilibrium. For a discussion on market equilibrium under a constant return technology and perfect competition, see, for example, \citet[pp.~192-193, pp.~217-218]{HayashiMicro2021}}

To satisfy this demand, each firm in region $x$ has to ship $c(y)T(x,y)$ units of the variety. Therefore, the total sales of each firm in region $x$ is given by $q(x)=\int_\Omega c(y)T(x,y)dy$, which is given from \eqref{cs} as
\begin{equation}\label{qr}
q(x) = \mu \int_\Omega Y(y)\left(p(x) T(x,y)\right)^{-\sigma} G(y)^{\sigma-1}T(x,y)dy.
\end{equation}

\subsection{Producer behavior}
The density of manufacturing firms in region $x$ is given by $Nn(x)$. Each firm produces $q(x)$ of one variety, and $q(x)$ is given by \eqref{qr}, hence the total output of manufactured goods in region $x$ is
\begin{align}
Nn(x) q(x) &= \mu  Nn(x) \left(p(x)\right)^{-\sigma}\int_\Omega Y(y)G(y)^{\sigma-1}T(x,y)^{1-\sigma}dy.\label{nrqr}
\end{align}
It is assumed, as in \citet{PicaThiToul04}, that when each firm produces $q$ quantities of a variety of manufactured goods, the required labor input $l$ is given by 
\begin{equation}\label{tech}
l = c_M q,
\end{equation}
where $c_M>0$ stands for a marginal cost. From \eqref{tech}, the total demand for manufacturing labor in region $x$ is $c_MNn(x)q(x)$, while the labor supply in region $x$ is $Mm(x)$. Thus, labor market clearing is achieved when
\begin{equation}\label{clearing}
c_M Nn(x)q(x) = Mm(x).
\end{equation}
It immediately follows from \eqref{clearing} that
\begin{equation}\label{qrmrbunnnonr}
q(x) = \frac{Mm(x)}{c_M Nn(x)}.
\end{equation}
The nominal profit of each firm in region $x$ is 
\begin{equation}\label{profit}
g(x) = (p(x) - c_Mw(x))q(x).
\end{equation}
Using \eqref{qr}, we obtain the price that maximizes \eqref{profit} as
\begin{equation}\label{pricing}
p(x) = \frac{\sigma}{\sigma-1}c_Mw(x).
\end{equation}
Applying \eqref{qrmrbunnnonr} and \eqref{pricing} to \eqref{profit}, we have
\begin{equation}
g(x) = \frac{w(x)}{\sigma-1} \frac{Mm(x)}{Nn(x)}.\label{maximizedprofit}
\end{equation}

\subsection{Model equations}\label{subsec:modeleqs}
Substituting \eqref{maximizedprofit} into \eqref{totalincome}, we see that the total nominal income in region $x$ becomes
\begin{equation}\label{Yr}
Y(x) = \Phi\phi(x) + \frac{\sigma}{\sigma-1}w(x)Mm(x).
\end{equation}

By \eqref{nrqr} and \eqref{pricing}, we see that the total demand for the manufacturing labor in region $x$ given by $c_MNn(x)q(x)$ is
\begin{equation}\label{cMNnrqr}
\begin{aligned}
c_M Nn(x) q(x) &=
\mu c_M Nn(x) \left(\frac{\sigma}{\sigma-1}c_Mw(x)\right)^{-\sigma}\\
&\hspace{25mm}\times \int_\Omega Y(y)G(y)^{\sigma-1}T(x,y)^{1-\sigma}dy.
\end{aligned}
\end{equation}
Putting \eqref{cMNnrqr} equals to the manufacturing labor supply $Mm(x)$ in region $x$, we obtain the nominal wage equation
\begin{equation}\label{nw}
w(x) = \frac{\sigma-1}{\sigma c_M}\left[\mu c_M \frac{Nn(x)}{Mm(x)}\int_\Omega Y(y)G(y)^{\sigma-1}T(x,y)^{1-\sigma}dy\right]^{\frac 1 \sigma}.
\end{equation}
This equation can be interpreted as follows. The closer a firm is to a large market (where $Y(y)$ is large), the more it can afford to pay higher wages.\footnote{Here, the meaning of ``close'' for regions $x$ and $y$ includes the perspective of transport costs, meaning that the value of $T(x,y)$ is small. This also applies to the following paragraph.} Also, the fewer other firms located close to the firm (which translates to high price indices due to reduced varieties), the higher the nominal wage the firm can afford to pay. As described in \citet[p.~53]{FujiKrugVenab}, these are the standard interpretations of the nominal wage equation of the CP model. Furthermore, the equation \eqref{nw} additionally incorporates the direct impact of labor market outcomes through the $\frac{n(x)}{m(x)}$ term. That is, when the number of firms in a region increases relatively, the nominal wage in that region rises; when the number of firms decreases relatively, the nominal wage in that region falls.

From \eqref{priceindex_r} and \eqref{pricing}, we obtain an equation for the price index as
\begin{equation}\label{pri}
G(x) = \frac{\sigma c_M}{\sigma-1}\left[\int_\Omega Nn(y) w(y)^{1-\sigma}T(x,y)^{1-\sigma}dy\right]^{\frac{1}{1-\sigma}}.
\end{equation}
Since the prices of individual varieties and nominal wages are linked by the markup relationship in equation \eqref{pricing}, it is natural that the price index---the weighted average of individual prices---is expressed by the equation \eqref{pri}. As \citet[p.~63]{FujiKrugVenab} state, this equation describes the forward linkages: the more firms agglomerate close to a given region, the lower the price index there becomes due to increased variety, making that region more attractive to consumers.

Maximizing \eqref{utility} under the budget constraint \eqref{budget_composite} yields 
\begin{equation}\label{CandA}
\begin{aligned}
&C = \mu YG^{-1},\\
&A = (1-\mu)Y.
\end{aligned}
\end{equation}
Substituting \eqref{CandA} into \eqref{utility}, we have the indirect utility function
\[
V = \mu^\mu(1-\mu)^{1-\mu} YG^{-\mu},
\]
which allows us to define real income by $YG^{-\mu}$. Hence, the real wage equation and the real profit equation in region $x$ are defined by 
\begin{equation}\label{rw}
\omega(x)=w(x)G(x)^{-\mu}
\end{equation}
and
\begin{equation}\label{rp}
\eta(x)=g(x)G(x)^{-\mu},
\end{equation}
respectively. To be precise, the equation \eqref{rw} refers to manufacturing real wages. However, in our model, only manufacturing workers are concerned with real wages, so we simply refer to them as real wages. By \eqref{maximizedprofit}, \eqref{rw}, and \eqref{rp}, we have
\begin{equation}\label{rp2}
\eta(x)  = \frac{1}{\sigma-1} \frac{Mm(x)}{Nn(x)}\omega(x).
\end{equation}
This equation states that real wages and real profits are proportional. This is a direct consequence of the proportional relationship between nominal profits and nominal wages \eqref{maximizedprofit}. This can be interpreted as a causal relationship: in regions with high wages, manufacturing workers' incomes are high, and consequently, profits are also high due to larger sales. Similar to the nominal wage equation \eqref{nw}, this equation reflects the direct impact of labor market outcomes through the $\frac{m(x)}{n(x)}$ term. That is, when the number of firms in a region increases relatively, the real profit in the region falls (while the real wage in the region rises); when the number of firms decreases relatively, the real profit in the region increases (while the real wage in the region falls).

The firms exit from regions with lower-than-average real profits and enter regions with higher-than-average real profits. The average real profit is defined by $\int_\Omega \eta(y)n(y)dy$. With the time variable denoted as $t$, the migration process is expressed by the replicator equation
\begin{equation}\label{dynn}
\frac{\partial n}{\partial t}(t, x) = v_n\left[\eta(t, x) - \int_\Omega\eta(t, y)n(t, y)\right]n(t, x),
\end{equation}
where $v_n > 0$ is the sensitivity of firms to disparities of real profits. Similarly, the average real wage is defined by $\int_\Omega\omega(y)m(y)dy$, and the following replicator equation is assumed for the migration process of manufacturing workers
\begin{equation}\label{dynm}
\frac{\partial m}{\partial t}(t, x) = v_m\left[\omega(t, x) - \int_\Omega\omega(t, y)m(t, y)\right]m(t, x),
\end{equation}
where $v_m > 0$ is the sensitivity of manufacturing workers to disparities in real wages.\footnote{Such replicator dynamics is also employed in \citet[p.62]{FujiKrugVenab}.}

\noindent

\subsection{Normalizations}
Let us normalize some units of measurement to simplify the model equations. We are free to choose units for the quantity of output, the number of the two types of workers, and the number of firms.

First, we choose a unit for the quantity of output so that the marginal cost in \eqref{tech} satisfies that\footnote{\citet[p.54]{FujiKrugVenab}}
\begin{equation}\label{cMsigm1sig}
c_M = \frac{\sigma-1}{\sigma}.
\end{equation}
Next, we choose units for the number of agricultural and manufacturing workers so that 
\begin{align}
&\Phi = (1-\mu)\sigma,\label{Phi1mmusigma}\\
&M =  \mu(\sigma-1) \label{Mmusigmam1}
\end{align}
hold, respectively. Finally, we choose a unit for the number of firms so that
\begin{equation}\label{Nsigma}
N = \sigma
\end{equation}
holds.

The model equations \eqref{Yr}, \eqref{nw}, \eqref{pri}, \eqref{rw}, \eqref{rp2}, \eqref{dynn}, and \eqref{dynm} with the normalizations \eqref{cMsigm1sig}, \eqref{Phi1mmusigma}, \eqref{Mmusigmam1}, and \eqref{Nsigma} constitute the following system of integral and differential equations:
\begin{equation}\label{0}
\left\{
\begin{aligned}
&Y(t,x) = (1-\mu)\sigma\phi(x) + \mu\sigma w(t,x)m(t,x),\\
&w(t,x) = \left[\frac{n(t,x)}{m(t,x)}\int_\Omega Y(t,y)G(t,y)^{\sigma-1}T(x,y)^{1-\sigma}dy\right]^{\frac 1 \sigma},\\ 
&G(t,x) = \left[\sigma\int_\Omega n(t,y) w(t,y)^{1-\sigma}T(x,y)^{1-\sigma}dy\right]^{\frac{1}{1-\sigma}},\\
&\eta(t,x) = \frac{\mu}{\sigma}\frac{m(t,x)}{n(t,x)}\omega(t,x),\\
&\omega(t,x) = w(t,x)G(t,x)^{-\mu},\\
&\frac{\partial n}{\partial t}(t,x) = v_n\left[\eta(t,x) - \int_\Omega\eta(t,y)n(t,y)dy\right]n(t,x),\\
&\frac{\partial m}{\partial t}(t,x) = v_m\left[\omega(t,x) - \int_\Omega\omega(t,y)m(t,y)dy\right]m(t,x),
\end{aligned}\right.
\end{equation}
with initial conditions $n(0, x)=n_{0}(x)\geq 0$ and $m(0, x)=m_{0}(x)\geq 0$ for all $x\in \Omega$. Here, it is made explicit that all functions other than $\phi$ depend on the time variable $t\geq 0$.

\subsection{Racetrack economy}\label{subsec:functionsonS}
In the following, we consider in particular the continuous space model \eqref{0} with $\Omega=S$, that is, a circle of radius $\rho>0$. The model is what \citet[p.82]{FujiKrugVenab} call the {\it racetrack economy}. Let $d(x,y)$ denote the shorter distance between any two points $x$ and $y$ in $S$. Then, we assume that
\[
T(x,y) = e^{\tau d(x,y)},\hspace{3mm}\forall x,y\in S,
\]
where the parameter $\tau\geq 0$ is referred to as a transport coefficient. In addition, it is useful to introduce a {\it generalized transport coefficient} $\alpha$ as
\begin{equation}\label{gtc}
\alpha := (\sigma-1)\tau \geq 0.
\end{equation}
because the parameters $\sigma$ and $\tau$ often appear in the above form.

It is convenient to identify functions on $S$ with periodic functions on $[-\pi, \pi]$ for the sake of concrete computations as in \citet{Ohtake2023cont} and \citet{Ohtake2025agriculture}. In fact, a point $x\in S$ corresponds one-to-one to an angle $\theta \in [-\pi,\pi)$ in the circular coordinate. Thus, we can write $x=x(\theta)$, which always allows this identification. Therefore, we allow $f(x),~x\in S$ to be written as $f(\theta)$ for $\theta\in[-\pi,\pi)$ and $f(\pi)=f(-\pi)$. By this identification, we can compute the integral on $S$ as
\[
\int_S f(x) dx = \int_{-\pi}^\pi f(\theta) \rho d\theta,
\]
and the shorter distance between $x(\theta)$ and $y(\theta^\prime)$ along $S$ as
\[
d(x,y) = \rho \min\left\{|\theta-\theta^\prime|, 2\pi-|\theta-\theta^\prime|\right\}.
\]

\section{Stability of homogeneous stationary solution}

\subsection{Homogeneous stationary solution}
Assuming that the agricultural workers are uniformly distributed as
\begin{equation}
\phi(x) \equiv \ol{\phi} = \frac{1}{2\pi \rho},~\forall x\in S, \label{phib}
\end{equation}
we obtain a {\it homogeneous stationary solution} in which all other economic factors are also uniformly distributed;
\begin{align}
&n(x) \equiv \ol{n} = \frac{1}{2\pi \rho},&\forall x\in S, \label{hn}\\ 
&m(x) \equiv \ol{m} = \frac{1}{2\pi \rho},&\forall x\in S, \label{hm}\\
&Y(x) \equiv \ol{Y} = \frac{\sigma}{2\pi\rho},&\forall x\in S, \label{hY}\\
&w(x) \equiv \ol{w} = 1,&\forall x\in S, \label{hw}\\
&G(x) \equiv \ol{G} = \left[\sigma\ol{n}\frac{2\left(1-e^{-\alpha\pi \rho}\right)}{\alpha}\right]^{\frac{1}{1-\sigma}},~&\forall x\in S, \label{hG}\\
&\eta(x) \equiv \ol{\eta} =\frac{\mu}{\sigma}\ol{\omega},&\forall x\in S, \\
&\omega(x) \equiv \ol{\omega} = \ol{G}^{-\mu},&\forall x\in S. \label{homega} 
\end{align}
These values are easily obtained by simply placing all the factors as constants in the first five equations in \eqref{0} and further noting that
\[
\int_S T(x,y)^{1-\sigma}dy=\int_Se^{-\alpha d(x,y)}dy
= \frac{2\left(1-e^{-\alpha \rho\pi}\right)}{\alpha}, 
\]
which no longer depends on $x$ in \eqref{hG}.

\subsection{Linearized problem}

Let small perturbations added to the homogeneous share functions $n$ and $m$ be $\Delta n$ and $\Delta m$, respectively. They must satisfy
\[
\int_S \Delta n(t,x)dx=0~\text{and}~\int_S \Delta m(t,x)dx=0,
\]
because of the conservation \eqref{conservation}. These perturbations $\Delta n$ and $\Delta m$ induce other perturbations $\Delta Y$, $\Delta w$, $\Delta G$, $\Delta\eta$, and $\Delta \omega$ to the solutions of the first five equations of \eqref{0}.

Substituting $n(t,x)=\ol{n}+\Delta n(t,x)$, $m(t,x)=\ol{m}+\Delta m(t,x)$, $Y(t,x)=\ol{Y}+\Delta Y(t,x)$, $w(t,x)=\ol{w}+\Delta w(t,x)$, $G(t,x)=\ol{G}+\Delta G(t,x)$, $\eta(t,x)=\ol{\eta}+\Delta\eta(t,x)$, and $\omega(t,x)=\ol{\omega}+\Delta\omega(t,x)$ into \eqref{0}, and neglecting second and higher order terms of the perturbations, we obtain the linearized system
\begin{equation}\label{L}
\left\{
\begin{aligned}
&\Delta Y(t, x) = \mu\sigma\ol{w}\Delta m(t, x) + \mu\sigma\ol{m}\Delta w(t, x),\\
&\Delta w(t,x) = \frac{\sigma-1}{\sigma}\frac{\ol{n}}{\ol{m}}\ol{Y}\ol{G}^{\sigma-2}\int_S\Delta G(t,y)e^{-\alpha d(x,y)}dy\\
&\hspace{35mm} + \frac{\ol{n}}{\ol{m}}\frac{\ol{G}^{\sigma-1}}{\sigma}\int_S\Delta Y(t,y)e^{-\alpha d(x,y)}dy \\
&\hspace{46.5mm} -\frac{1}{\sigma\ol{m}^22\pi\rho}\Delta m(t,x) + \frac{1}{\sigma\ol{m}\ol{n}2\pi\rho}\Delta n(t,x),\\ 
&\Delta G(t,x) = \sigma\ol{G}^\sigma\ol{n}\ol{w}^{-\sigma}\int_S\Delta w(t,y)e^{-\alpha d(x,y)}dy\\
&\hspace{23mm} + \frac{\sigma}{1-\sigma}\ol{G}^\sigma\ol{w}^{1-\sigma}\int_S\Delta n(t,y)e^{-\alpha d(x,y)}dy,\\
&\Delta \eta(t,x) = \frac{\mu}{\sigma}\frac{\ol{m}}{\ol{n}}\Delta\omega(t,x) - \frac{\mu}{\sigma}\frac{\ol{m}}{\ol{n}^2}\ol{\omega}\Delta n(t,x) + \frac{\mu}{\sigma}\frac{\ol{\omega}}{\ol{n}}\Delta m(t,x),\\
&\Delta \omega(t,x) = \ol{G}^{-\mu}\Delta w(t,x) -\mu\ol{w}\ol{G}^{-\mu-1}\Delta G(x),\\
&\frac{d\Delta n}{dt}(t,x) = v_n \ol{n}\Delta \eta(t,x),\\
&\frac{d\Delta m}{dt}(t,x) = v_m \ol{m}\Delta \omega(t,x).
\end{aligned}\right.
\end{equation}

Let us expand these small perturbations on $[0,\infty)\times S$ identified with the corresponding functions on $[0,\infty)\times [-\pi, \pi)$ into the Fourier series concerning the space variable. The Fourier series of a perturbation $\Delta f(t,\theta)$ is defined as
\[
\Delta f(t,\theta) = \frac{1}{2\pi}\sum_{k=0,\pm1,\pm2,\cdots}\hat{f}_k(t) e^{\sqrt{-1}k\theta},
\]
where $e^{\sqrt{-1}k\theta}$ is called the eigenmode with the $k$-th frequency (hereafter abbreviated as the $k$-th mode). The Fourier coefficient $\hat{f}_k$ of the $k$-th mode is given as 
\[
\hat{f}_k(t) = \int_{-\pi}^\pi \Delta f(t, \theta) e^{-\sqrt{-1}k\theta} d\theta
\]
for $k=0,\pm1,\pm2,\cdots$. Substituting the Fourier series of $\Delta n$, $\Delta m$, $\Delta Y$, $\Delta w$, $\Delta G$, $\Delta \eta$, and $\Delta \omega$ into the linearized system \eqref{L}, and using \eqref{hn}, \eqref{hm}, \eqref{hY}, \eqref{hw}, and \eqref{homega}, we obtain the following system of linear equations for the Fourier coefficients
\begin{equation}\label{fc}
\left\{
\begin{aligned}
&\hat{Y}_k = \mu\sigma\hat{m}_k + \mu\sigma\ol{m}\hat{w}_k,\\
&\hat{w}_k = \frac{\sigma-1}{\sigma}\ol{G}^{-1}Z_k\hat{G}_k
+ \frac{1}{\sigma^2\ol{n}}Z_k\hat{Y}_k
-\frac{1}{\sigma\ol{m}^22\pi\rho}\hat{m}_k
+ \frac{1}{\sigma\ol{m}\ol{n}2\pi\rho}\hat{n}_k,\\
&\hat{G}_k = \ol{G} Z_k \hat{w}_k + \frac{\ol{G}}{(1-\sigma)\ol{n}}Z_k\hat{n}_k,\\
&\hat{\eta}_k = \frac{\mu}{\sigma}\hat{\omega}_k - \frac{\mu}{\sigma}\frac{\ol{G}^{-\mu}}{\ol{n}}\hat{n}_k + \frac{\mu}{\sigma}\frac{\ol{G}^{-\mu}}{\ol{n}}\hat{m}_k,\\
&\hat{\omega}_k = \ol{G}^{-\mu}\hat{w}_k - \mu\ol{G}^{-\mu-1}\hat{G}_k,\\
&\frac{d}{dt}\hat{n}_k = v_n\ol{n}\hat{\eta}_k,\\
&\frac{d}{dt}\hat{m}_k = v_m\ol{m}\hat{\omega}_k,
\end{aligned}\right.
\end{equation}
where
\begin{equation}\label{Zk}
Z_k := \frac{\alpha^2\rho^2\left(1-(-1)^ke^{-\alpha \rho\pi}\right)}{\left(k^2+\alpha^2\rho^2\right)\left(1-e^{-\alpha \rho\pi}\right)},\hspace{3mm}k=\pm1,\pm2,\cdots
\end{equation}
Note that the case $k=0$ needs not be considered because $\hat{n}_0=\hat{m}_0=0$. The variable $Z_k$ should be called a {\it transport index} which plays a vital role in the analysis of the linearized system. The transport index $Z_k$ emerges as the result of
\[
\begin{aligned}
\ol{G}^{\sigma-1}\int_{-\pi}^{\pi}e^{\sqrt{-1}k\theta^\prime}e^{-\alpha \rho\left|\theta-\theta^\prime\right|}\rho d\theta^\prime 
&= \frac{\alpha}{2\sigma\ol{n}\left(1-e^{-\alpha \rho\pi}\right)}
\frac{2\alpha\rho^2\left(1-(-1)^ke^{-\alpha \rho\pi}\right)}{k^2+\alpha^2\rho^2}e^{\sqrt{-1}k\theta} \\
&= \frac{1}{\sigma\ol{n}}Z_k e^{\sqrt{-1}k\theta},
\end{aligned}
\]
where we also use \eqref{hG}. The next proposition about the transport index is quite useful.\footnote{The transport index $Z_k$ plays essentially the same role as the variable designated as $Z$ in \citet[Chapters 4, 5, and 6]{FujiKrugVenab}. They call $Z$ ``{\it a sort of index of trade cost}'' (p.~57). This variable is also useful in analyzing other NEG models on a one-dimensional circle, as in \citet{Ohtake2023cont} and \citet{Ohtake2025agriculture}.} For the proof, see Subsection \ref{subsec:thZ01}.
{\vspace{2mm}}
\begin{prop}\label{th:Z01}
The transport index $Z_k$ defined by \eqref{Zk} is monotonically increasing in $\alpha \rho\geq 0$. Specifically,
\[
\begin{aligned}
&\lim_{\alpha \rho \to 0} Z_k = 0,\\
&\lim_{\alpha \rho \to \infty} Z_k = 1,
\end{aligned}
\]
hold.
\end{prop}
In \eqref{fc}, substituting the first equation into the second one, and using \eqref{hn} and \eqref{hm}, we can express the second and the third equations as

\begin{equation}\label{simeqwhatGhat}
\begin{bmatrix}
1-\frac{\mu}{\sigma}Z_k & -\frac{\sigma-1}{\sigma}Z_k \\[3mm]
-Z_k & 1
\end{bmatrix}
\begin{bmatrix}
\hat{w}_k \\[3mm]
\frac{\hat{G}_k}{\ol{G}}
\end{bmatrix}
=
2\pi \rho
\begin{bmatrix}
\frac{1}{\sigma}\hat{n}_k - \frac{1}{\sigma}\left(1-\mu Z_k\right)\hat{m}_k\\[3mm]
-\frac{1}{\sigma-1}Z_k\hat{n}_k.
\end{bmatrix}
\end{equation}
Let $\delta(Z_k)$ be the determinant of the $2\times 2$ matrix in the left-hand side of \eqref{simeqwhatGhat}:
\begin{equation}\label{deltaZk}
\delta(Z_k) = 1-\frac{\mu}{\sigma}Z_k -\frac{\sigma-1}{\sigma}Z_k^2.
\end{equation}
It is easily verified that
\begin{equation}\label{delta>0}
\delta(Z) >0
\end{equation}
for all $Z\in [0,1]$ as a quadratic function of $Z$. Then, we can solve \eqref{simeqwhatGhat} as
\begin{equation}\label{slvdwhGh}
\begin{aligned}
\begin{bmatrix}
\hat{w}_k \\[3mm]
\frac{\hat{G}_k}{\ol{G}}
\end{bmatrix}
&= \frac{2\pi \rho}{\sigma\delta(Z_k)}
\begin{bmatrix}
\left(1 - Z_k^2\right)\hat{n}_k - \left(1 - \mu Z_k\right)\hat{m}_k\\[3mm]
\left(-\frac{Z_k}{\sigma-1} + \frac{\mu Z_k^2}{\sigma-1}\right)\hat{n}_k - \left(Z_k - \mu Z_k^2\right)\hat{m}_k
\end{bmatrix}
\end{aligned}
\end{equation}
Substituting \eqref{slvdwhGh} into the fifth equation of \eqref{fc}, we have
\begin{equation}\label{slvdomh}
\hat{\omega}_k 
= \frac{2\pi \rho\ol{G}^{-\mu}}{\sigma}\left[A\hat{n}_k + B\hat{m}_k\right],
\end{equation}
where
\begin{align}
A&:= 
 \frac{1}{\delta(Z_k)} \left[1 + \frac{\mu}{\sigma-1}Z_k - \left(1+\frac{\mu^2}{\sigma-1}\right)Z_k^2\right], \label{defA}\\
B&:= 
\frac{1}{\delta(Z_k)} \left[-\left(\mu Z_k-1\right)^2\right]. \label{defB}
\end{align}
Since \eqref{delta>0} holds, it is easy from \eqref{defA} and \eqref{defB} to see that 
\begin{equation}\label{B<0}
A > 0,~B< 0
\end{equation}
for any $Z\in[0, 1]$. Substituting \eqref{slvdomh} into the fourth equation of \eqref{fc}, we have
\begin{equation}\label{slvdeta}
\hat{\eta}_k = \frac{\mu 2\pi \rho}{\sigma}\ol{G}^{-\mu} \left[\left(\frac{A}{\sigma}-1\right)\hat{n}_k + \left(\frac{B}{\sigma} + 1\right)\hat{m}_k\right].
\end{equation}
By \eqref{slvdomh} and \eqref{slvdeta}, we see that the sixth and seventh equations in \eqref{fc} become
\[
\begin{aligned}
\frac{d}{dt}
\left[
\begin{array}{c}
\hat{n}_k\\
\hat{m}_k
\end{array}
\right]
&=
\frac{\ol{G}^{-\mu}}{\sigma}\left[
\begin{array}{cc}
v_n\mu\left(\frac{A}{\sigma}-1\right) & v_n\mu\left(\frac{B}{\sigma}+1\right) \\[3mm]
v_m A & v_m B
\end{array}
\right] 
\left[
\begin{array}{c}
\hat{n}_k \\[3mm]
\hat{m}_k
\end{array}
\right].
\end{aligned}
\]

{\subsection{Analysis on eigenvalues}}

By the well-known result in the theory of dynamical systems,\footnote{For example, \citet[Chapter 5, p.92]{HirSm74}} if the real parts of all the eigenvalues of the matrix 
\begin{equation}\label{Mat}
\frac{\ol{G}^{-\mu}}{\sigma}
\left[
\begin{array}{cc}
v_n\mu\left(\frac{A}{\sigma}-1\right) & v_n\mu\left(\frac{B}{\sigma}+1\right) \\[3mm]
v_m A & v_m B
\end{array}
\right].
\end{equation}
are negative, then $n_k(t)\to 0$ and $m_k(t)\to 0$ as $t\to\infty$. Hence, the $k$-th mode is {\it stable}. On the other hand, if there is at least one eigenvalue having a positive real part, then $n_k(t)\to \infty$ or $m_k(t)\to \infty$, i.e., the $k$-th mode is {\it unstable}.

Figure \ref{fig:eigenvals} plots the maximal real part of the eigenvalues of \eqref{Mat} as a function of $\tau\geq 0$ for each frequency $k=1$, $2$, $3$, $4$, $5$, and $6$ under \eqref{nbh}.  The parameters are set to $\mu=0.6$, $\sigma=5.0$, $v_n=1.0$, $v_m=1.0$, and $\rho=1.0$. The Python code used to generate Figure \ref{fig:eigenvals} is available on GitHub \url{https://github.com/k-ohtake/dual-migration}.

As shown in Figure \ref{fig:eigenvals}, there exists a critical point $\tau^*_k$ where the maximal real part equals zero for each frequency $k$. It is easy to verify that the matrix \eqref{Mat} is the same for $k$ and $-k$ and thus $\tau^*_k=\tau^*_{-k}$ holds. The maximal real part of the eigenvalues takes positive values for $\tau\in(0, \tau^*_k)$, and negative values for $\tau\in(\tau^*_k,\infty)$. Therefore, for any frequencies $k$ and $k^\prime$ satisfying $|k|<|k^\prime|$, as transport costs decrease, the $k^\prime$-th mode becomes unstable before the $k$-th mode (For example, in Figure \ref{fig:eigenvals}, the sixth mode is unstable when $\tau=5$ but the fifth mode is still stable.).

\begin{figure}[H]
\centering
\includegraphics[width=10cm]{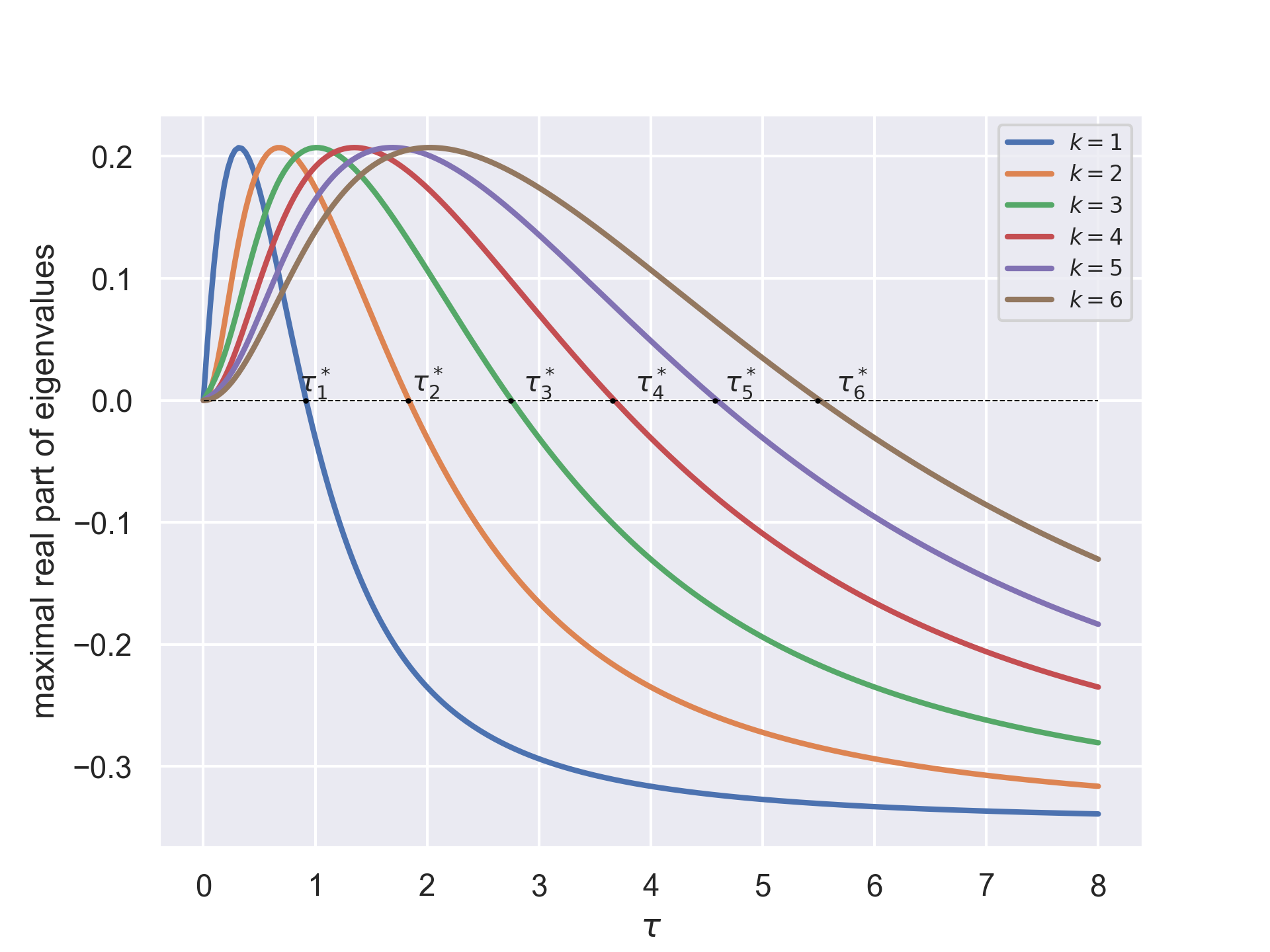}
\caption{The maximal real part of the eigenvalues of the matrix \protect\eqref{Mat} when parameter values are $\mu=0.6$, $\sigma=5.0$, $v_n=1.0$, $v_m=1.0$, and $\rho=1.0$.}\label{fig:eigenvals}
\end{figure}

The following proposition provides a sufficient condition for the critical point always to exist. See Subsection \ref{subsec:existence of cp} for the proof.

\begin{prop}\label{th: existence of critical point}
Suppose that 
\begin{equation}\label{nbh}
\frac{1}{1-\mu} < \sigma
\end{equation}
holds. Then, for any $k\neq 0$, there exists a critical point $\tau^*_k>0$.
\end{prop}
\noindent
{\bf Remark}: The condition \eqref{nbh} is so-called the {\it assumption of no black holes} in \citet[p.59]{FujiKrugVenab}.\footnote{In general, the specific formula varies depending on models.} It requires either that $\mu$ be sufficiently small or that $\sigma$ be sufficiently large. Therefore, this proposition implies that the emergence of extremely strong agglomeration forces---which makes any mode always unstable---can be suppressed if either the expenditure share of manufactured goods is small, or if the preference for variety is sufficiently weak. 

\vspace{5mm}
The following proposition states that the critical point $\tau_k^*$ increases linearly in even numbers of $|k|=2l$ ($l\in\mathbb{N}$). In other words, no matter how high transport costs are, modes with sufficiently large absolute frequencies become unstable (recall that the $k$-th mode is unstable on the interval $(0,\tau^*_k)$), and thus, the homogeneous stationary solution is always unstable in the sense that unstable eigenmodes always exist. See Subsection \ref{subsec:criticalpoint_kk+2} for the proof of the proposition.
\begin{prop}\label{th:criticalpoint_kk+2}
Suppose that \eqref{nbh} holds. For $|k|=2l$ where $l\in\mathbb{N}$, the critical point $\tau_k^*$ is given by
\begin{equation}\label{tauk*explicit}
\tau_k^* = \frac{|k|}{(\sigma-1)\rho}\sqrt{\frac{Z^*}{1-Z^*}}>0,
\end{equation}
where
\[
Z^* := \frac{\mu(2\sigma-1)}{\sigma(1+\mu^2)-1}\in(0, 1).
\]
\end{prop}

\noindent
{\bf Remark}: The homogeneous stationary solution is always unstable, but for modes with sufficiently large $|k|$, the maximum eigenvalue is sufficiently small that it can be considered practically negligible. In fact, when $|k|\to\infty$, it is easy from \eqref{deltaZk}, \eqref{defA}, and \eqref{defB} to verify that $\delta(Z_k)\to 1$, $A\to 1$, $B\to -1$, and the maximal eigenvalue of \eqref{Mat} converges to $0$ under any values of transport costs.\footnote{In fact, the maximum root of the characteristic polynomial \eqref{charapoly} converges to zero.}

\vspace{1mm}
\noindent
{\bf Remark}: When $|k|$ is odd, it is not possible to find an explicit expression  for $\tau_k^*$ as in \eqref{tauk*explicit}. Nevertheless, it can be shown that $\tau_k^*$ is an increasing function of $|k|$ and that $\lim_{|k|\to\infty}\tau_k^*=\infty$. See Subection \ref{tauk*explicit-oddk} for the proof.

\section{Asymptotic behavior of solutions}\label{sec:simu}

\subsection{Analytical facts on stationary solutions}\label{subsec:analytical}
Here, we show analytical facts that hold regardless of whether the stationary solution to \eqref{0} is homogeneous or non-homogeneous. Let us denote stationary solutions, which are not necessarily homogeneous, by appending a tilde to each state variable of the system \eqref{0}. For example, $\tilde{n}$, $\tilde{m}$, $\tilde{\eta}$, and $\tilde{\omega}$ stand for the share function of firms, the share function of manufacturing workers, the real profit, and the real wage in a stationary solution, respectively. In addition, let us introduce the support of any non-negative function $f$ on $S$ as
\[
{\rm supp}(f):=\left\{x\in S\mid f(x)>0 \right\}.
\]

Then, the next proposition ---which might be called {\it co-location}--- holds for any stationary solutions to \eqref{0}.
\begin{prop}\label{prop:colocation}
For $\tilde{n}$ and $\tilde{m}$ in any stationary solutions to \eqref{0},
\begin{equation}\label{suppneqsuppm}
{\rm supp}(\tilde{n})={\rm supp}(\tilde{m})
\end{equation}
holds. 
\end{prop}
Proposition \ref{prop:colocation} can be verified as follows. From \eqref{rp2}, we have
\[
\tilde{\eta}(x)  = \frac{1}{\sigma-1} \frac{M\tilde{m}(x)}{N\tilde{n}(x)}\tilde{\omega}(x),
\]
and we immediately see that if $\tilde{m}(x)= 0$ on $x\in {\rm supp}(\tilde{n})$ then $\tilde{\eta}(x)=0$ on $x\in {\rm supp}(\tilde{n})$. This implies that firms must leave regions lacking manufacturing workers. Therefore, this contradicts the fact that $\tilde{n}$ is in a steady state. Also from \eqref{rp2}, we obtain 
\[
\tilde{\omega}(x)=(\sigma-1)\frac{N\tilde{n}(x)}{M\tilde{m}(x)}\tilde{\eta}(x).
\]
Threfore, if $\tilde{n}(x)=0$ on $x\in {\rm supp}(\tilde{m})$ then $\tilde{\omega}(x)=0$ on $x\in {\rm supp}(\tilde{m})$. This implies that manufacturing workers must leave regions lacking firms. Therefore, this contradicts the fact that $\tilde{m}$ is in a steady state. Thus, we obtain \eqref{suppneqsuppm}.

Furthermore, the following proposition  ---which might be called {\it proportionality}--- holds for any stationary solutions to \eqref{0}. See Subsection \ref{proof:neqm} for the proof.
\begin{prop}\label{prop:neqm}
For $\tilde{n}$ and $\tilde{m}$ in any stationary solutions to \eqref{0},
\[
\tilde{n}(x)=\tilde{m}(x),\hspace{3mm}\forall x\in {\rm supp}(\tilde{n})={\rm supp}(\tilde{m})
\]
holds.
\end{prop}
\noindent
{\bf Remark}: Note that these propositions hold regardless of whether the stationary solution is homogeneous or non-homogeneous. Interestingly, the shapes of the distributions of $n$ and $m$ in stationary solutions are determined to be identical. As the proof makes clear, this is due to the dynamics being the replicator equation.

\subsection{Numerical examples}\label{sbsec:numerical}
We numerically investigate the asymptotic behavior of solutions of \eqref{0}. See Subsection \ref{subsec:numerical_scheme} for the numerical scheme. As discussed below, a solution starting near the homogeneous stationary solution eventually converges to a nonuniform stationary solution. As an initial condition, we add small perturbations to both share functions $n$ and $m$, i.e., 
\[\left\{\begin{aligned}
&n_0(x) = \ol{n} + \Delta n(x), \\
&m_0(x)=\ol{m} + \Delta m(x),
\end{aligned}\right.
\]
where the small perturbations $\Delta n$ and $\Delta m$ satisfy that $\int_S\Delta n(x)dx=\int_S\Delta m(x)dx=0$.\footnote{These initial functions are also discretized appropriately in the actual computation. } Each of the small perturbations is randomly generated for each simulation\footnote{Therefore, initial condition generally varies with each simulation.} and thus generally $\Delta n\neq \Delta m$. A time-evolving numerical solution stops moving numerically and reaches an approximated stationary solution after sufficient time. The transport coefficient $\tau>0$ is considered to be a control parameter, and other parameters are set to $\mu=0.6$, $\sigma=5.0$, $v_n=1.0$, $v_m=1.0$, and $\rho=1.0$.

The following figures \ref{fig:tau0p5}-\ref{fig:tau2p6} show approximated stationary solutions thus obtained.\footnote{In the figures, the actual computed values are indicated by the dots. The dashed lines are just interpolation.}  In each figure, the top left shows the share function of firms, the top right shows the distribution of real profits, the bottom left shows the share function of manufacturing workers, and the bottom right shows the distribution of real wages. The Julia code for the simulation is available on GitHub \url{https://github.com/k-ohtake/dual-migration}.

Spiky distributions are observed, with the share functions taking extremely high values in some regions and almost zero in others. It is also observed that regions with large agglomerations of firms and manufacturing workers, which should be called ``cities'', enjoy higher real profits and real wages than regions without such agglomerations. Furthermore, it is shown that as transport costs decrease, the number of cities decreases.\footnote{Depending on randomly varying initial values, the number of cities in stationary solutions also varies within a certain range. However, the maximum number of possible cities is controlled by the values of $\tau$. These figures show the distributions having the maximum number of cities for each value of $\tau$, obtained by simulating several times with generally different initial values.} This property is common to a lot of models based on the CP model.\footnote{For example, studies on two-regional models by \citet{ForsOtta} and \citet{Pfl}; multi-regional model by \citet{GasCasCorr}, \citet{AkaTakaIke}, \citet{IkeAkaKon}, and \citet{TabuThis}; one-dimensional continuous periodic model by \citet[Chapter 6]{FujiKrugVenab}, \citet{OhtakeYagi_point}, and  \citet{Ohtake2023cont}.}

\newpage
\newgeometry{top=25truemm, bottom=35truemm} 
\begin{figure}[H]
 \begin{subfigure}{0.5\columnwidth}
  \centering
  \includegraphics[width=\columnwidth]{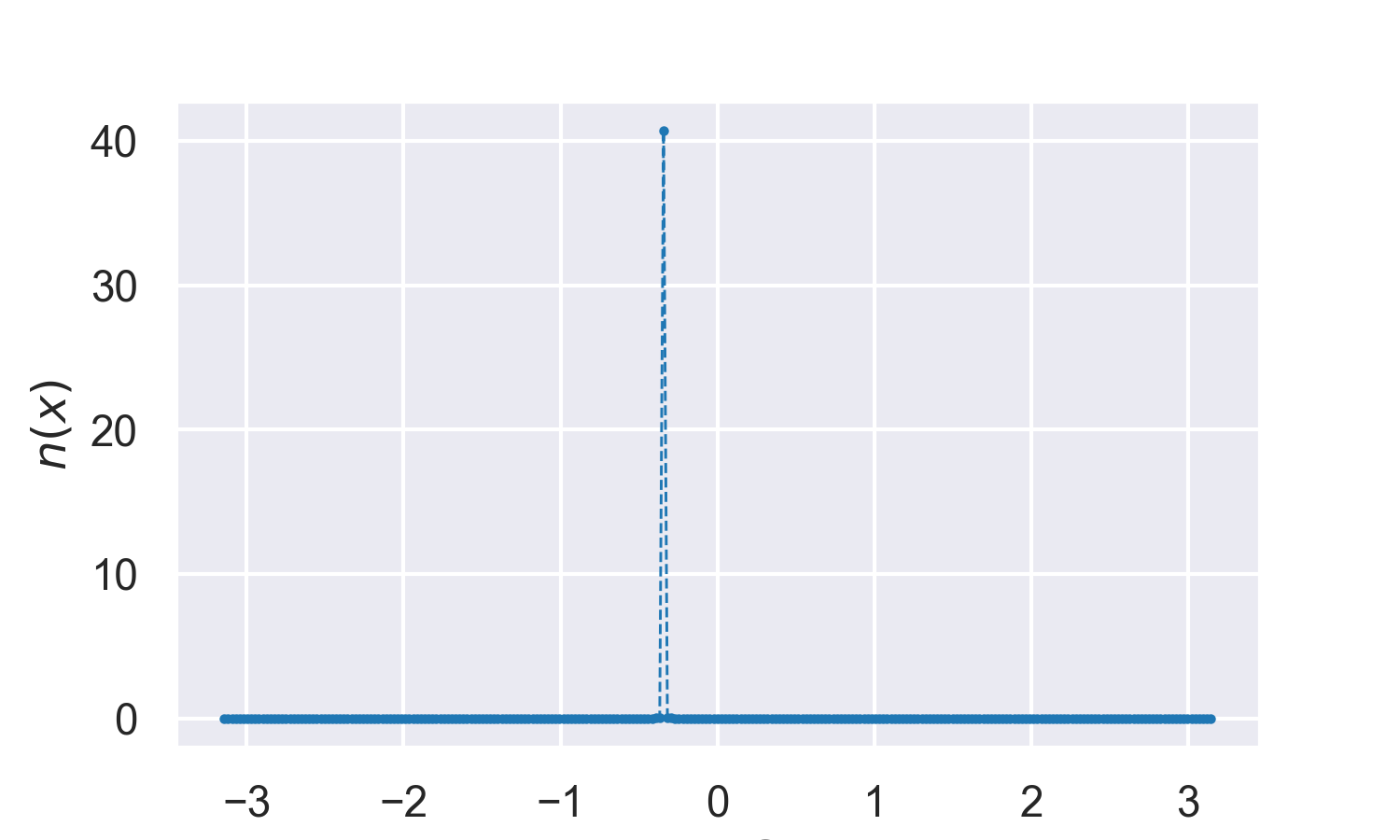}
  \caption{Firms}
 \end{subfigure}
 \begin{subfigure}{0.5\columnwidth}
  \centering
  \includegraphics[width=\columnwidth]{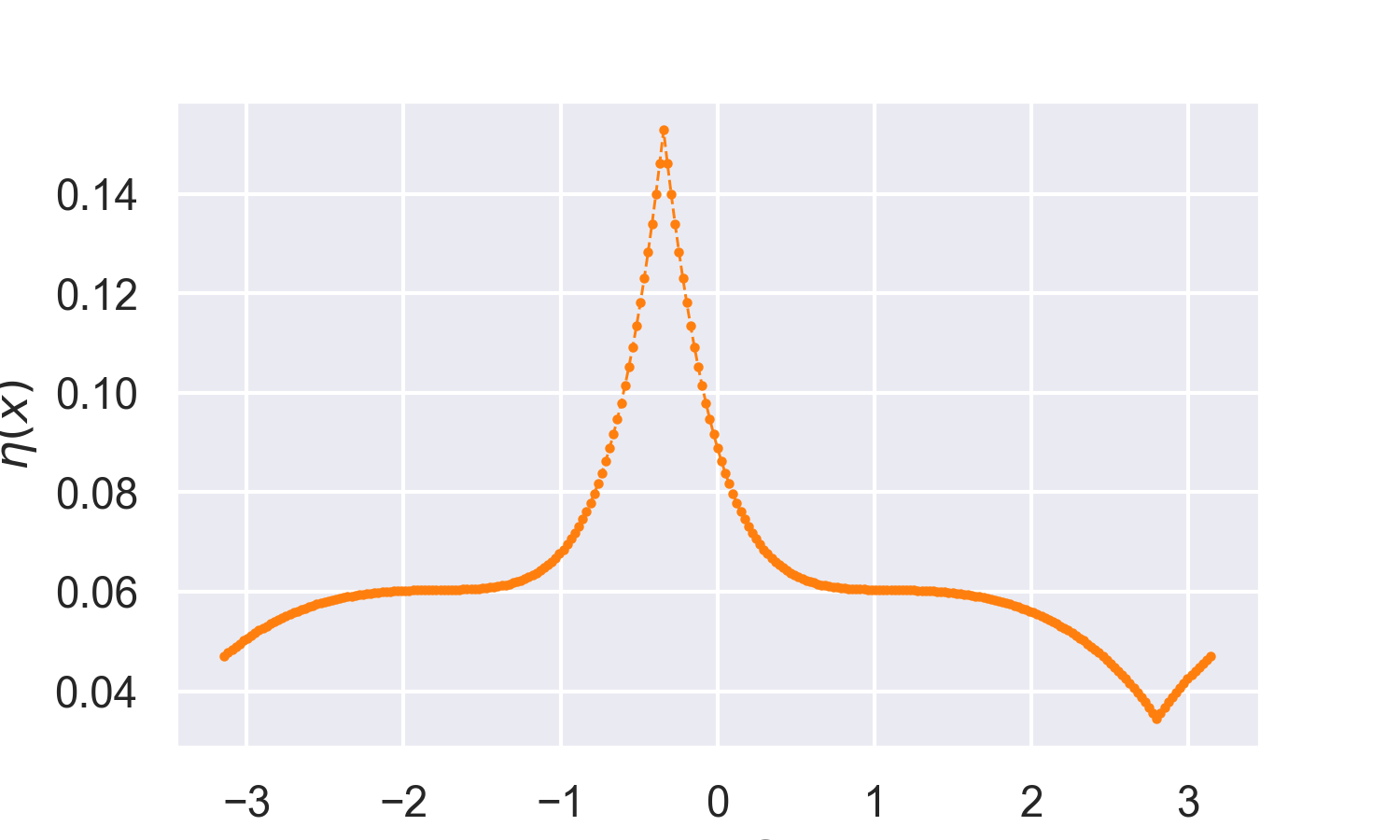}
  \caption{Real profits}
 \end{subfigure}
 \begin{subfigure}{0.5\columnwidth}
  \centering
  \includegraphics[width=\columnwidth]{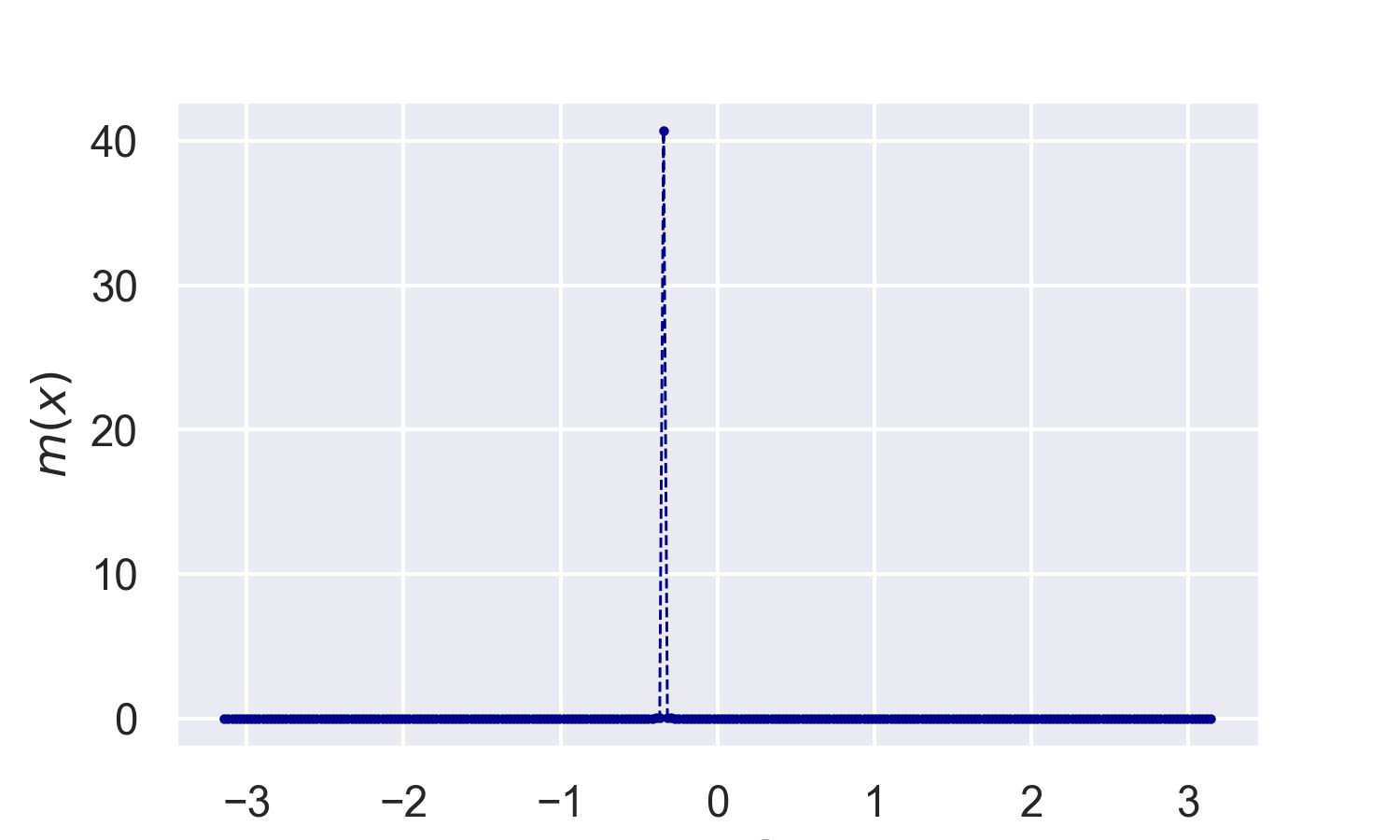}
  \caption{Manufacturing workers}
 \end{subfigure}
 \begin{subfigure}{0.5\columnwidth}
  \centering
  \includegraphics[width=\columnwidth]{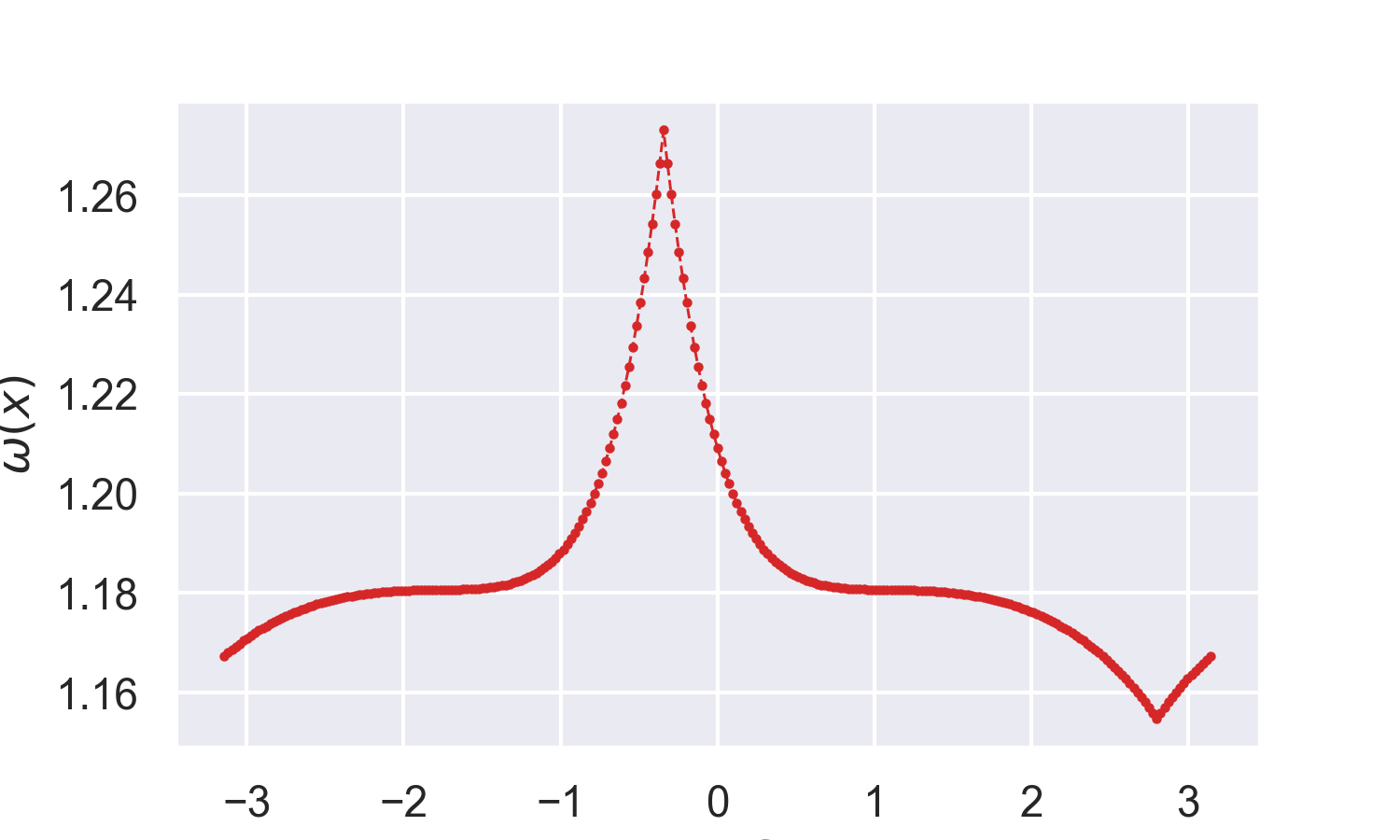}
  \caption{Real wages}
 \end{subfigure}
 \caption{A stationary solution for $\tau=0.5$ when other parameter values are $\mu=0.6$, $\sigma=5.0$, $v_n=1.0$, $v_m=1.0$, and $\rho=1.0$. A single city is formed.}
 \label{fig:tau0p5}
\end{figure}

\begin{figure}[H]
 \begin{subfigure}{0.5\columnwidth}
  \centering
  \includegraphics[width=\columnwidth]{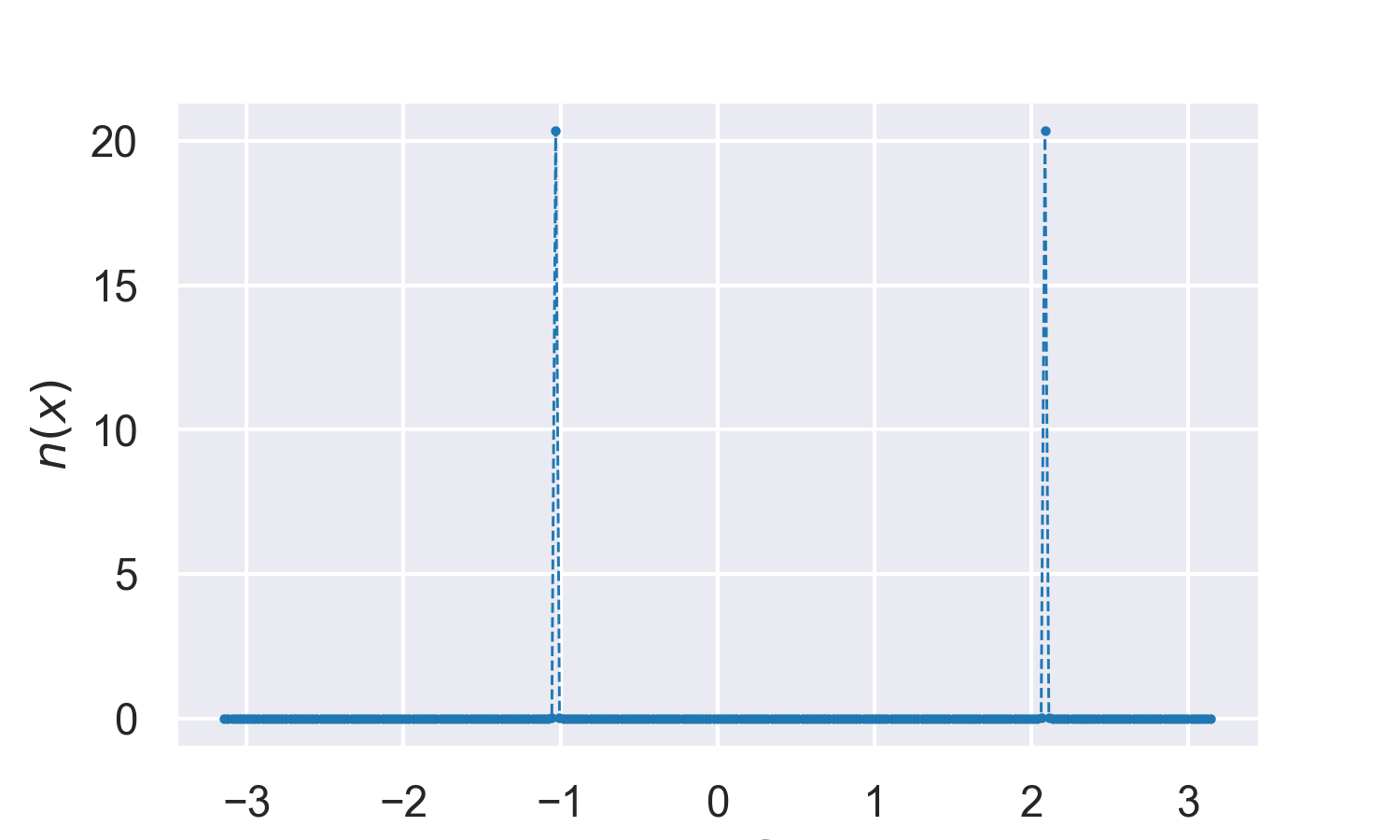}
  \caption{Firms}
 \end{subfigure}
 \begin{subfigure}{0.5\columnwidth}
  \centering
  \includegraphics[width=\columnwidth]{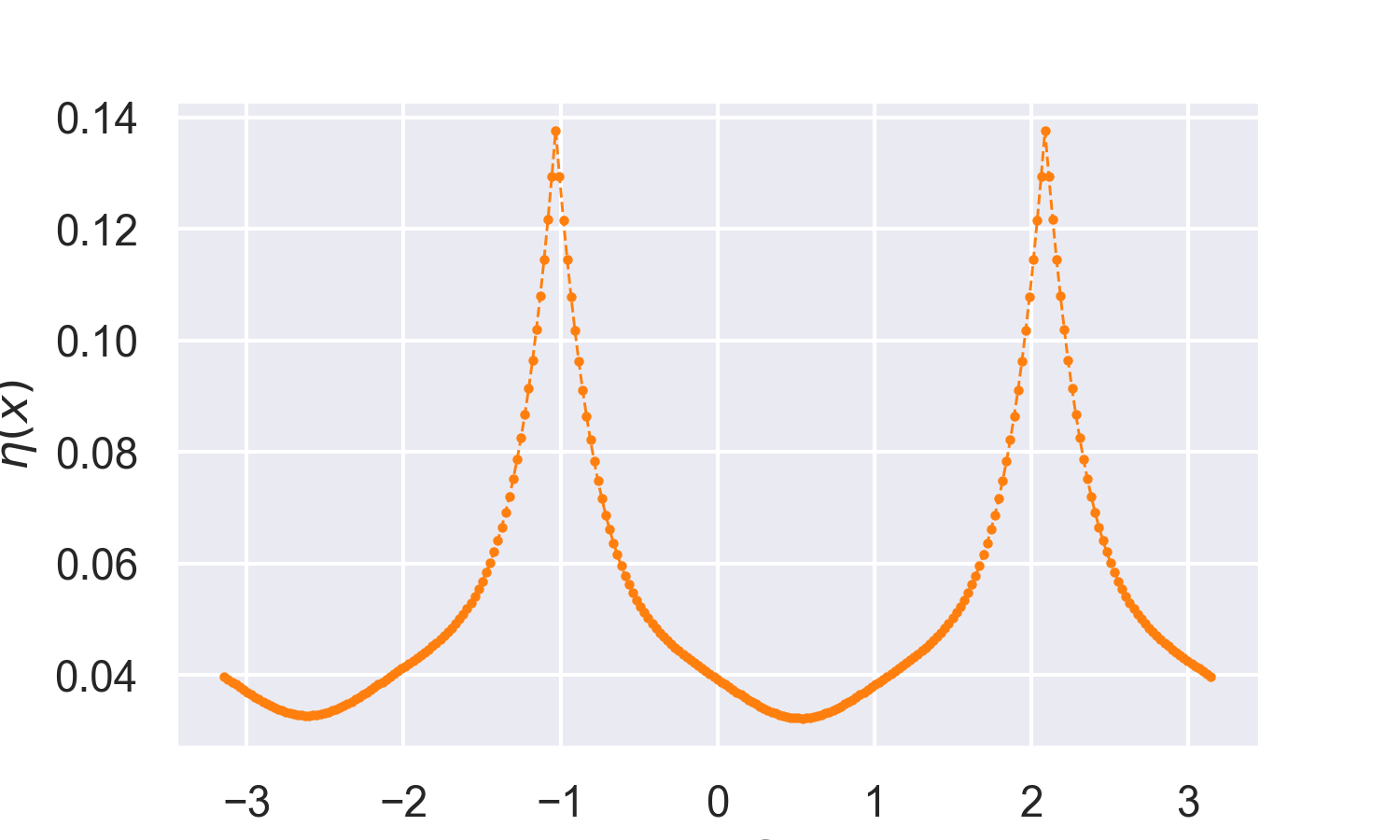}
  \caption{Real wages}
 \end{subfigure}
 \begin{subfigure}{0.5\columnwidth}
  \centering
  \includegraphics[width=\columnwidth]{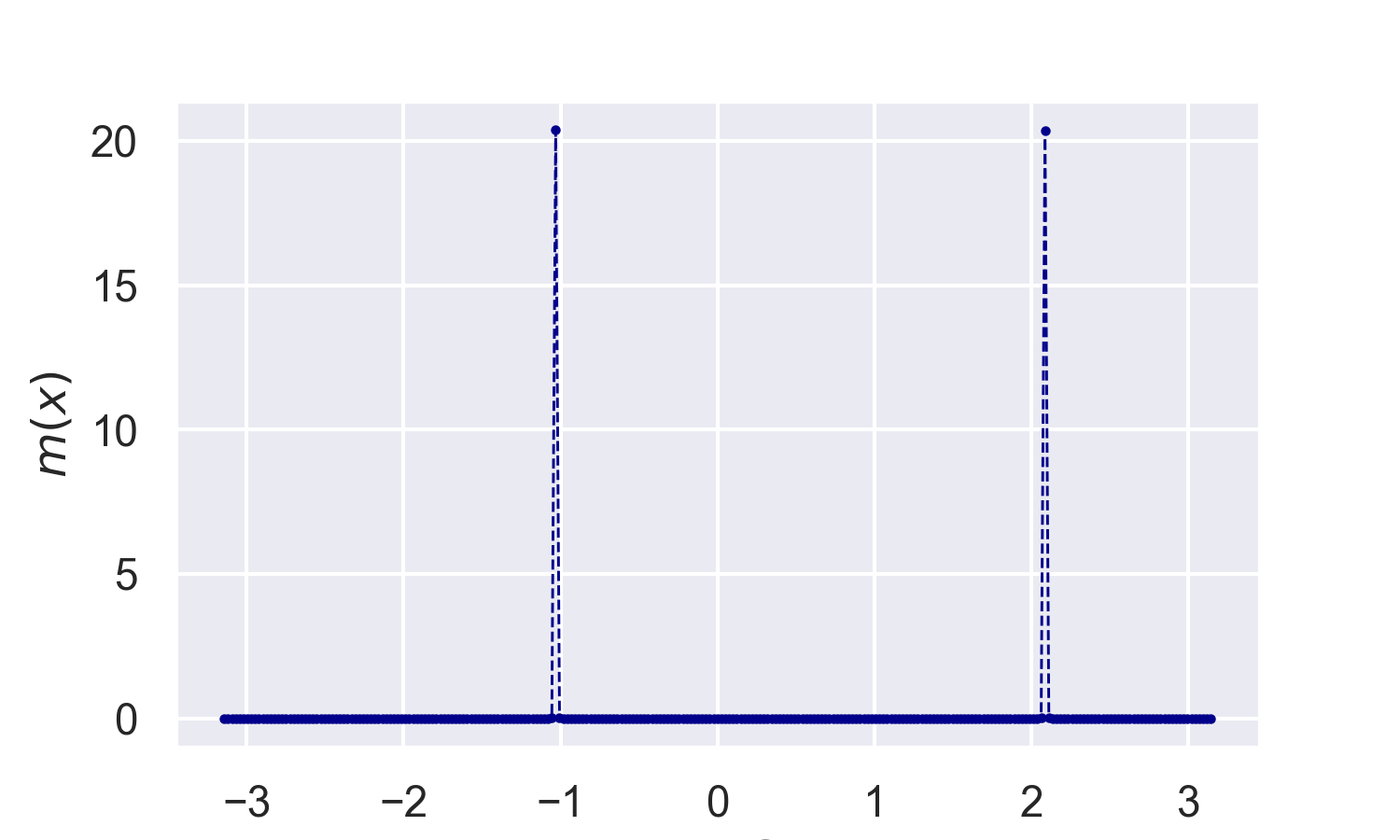}
  \caption{Manufacturing workers}
 \end{subfigure}
 \begin{subfigure}{0.5\columnwidth}
  \centering
  \includegraphics[width=\columnwidth]{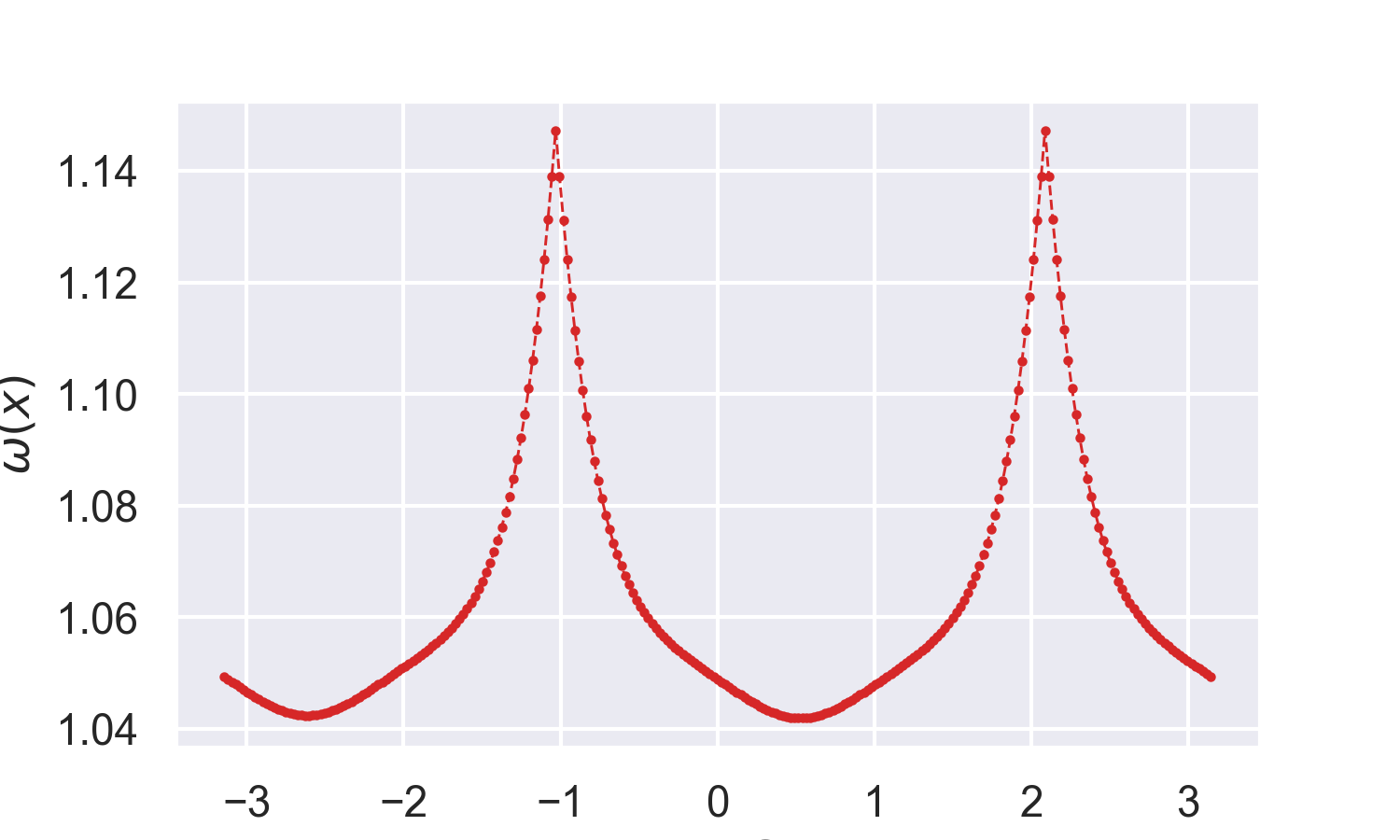}
  \caption{Real profits}
 \end{subfigure}
 \caption{A stationary solution for $\tau=0.7$ when other parameter values are $\mu=0.6$, $\sigma=5.0$, $v_n=1.0$, $v_m=1.0$, and $\rho=1.0$. Two cities are formed.} 
 \label{fig:tau0p7}
\end{figure}

\begin{figure}[H]
 \begin{subfigure}{0.5\columnwidth}
  \centering
  \includegraphics[width=\columnwidth]{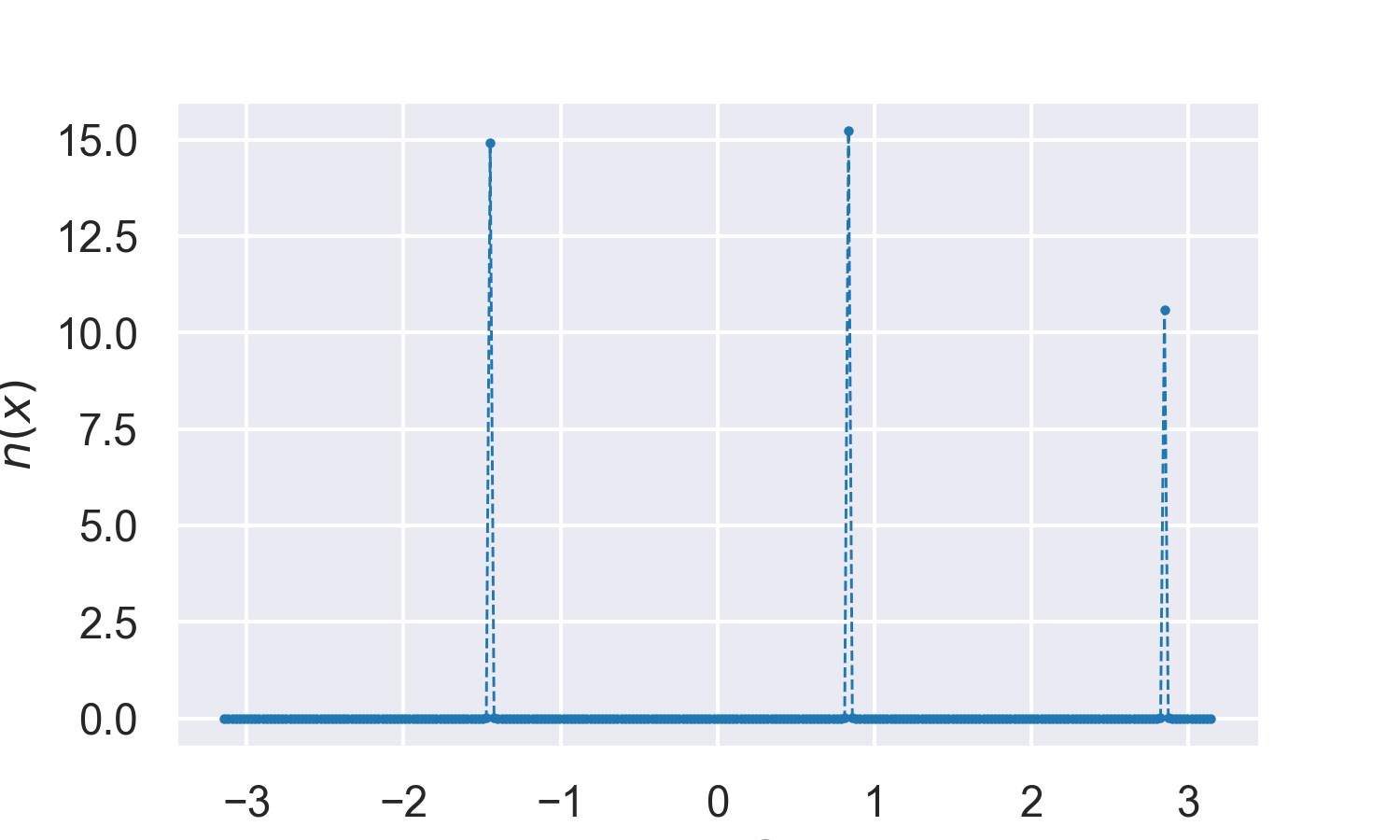}
  \caption{Firms}
 \end{subfigure}
 \begin{subfigure}{0.5\columnwidth}
  \centering
  \includegraphics[width=\columnwidth]{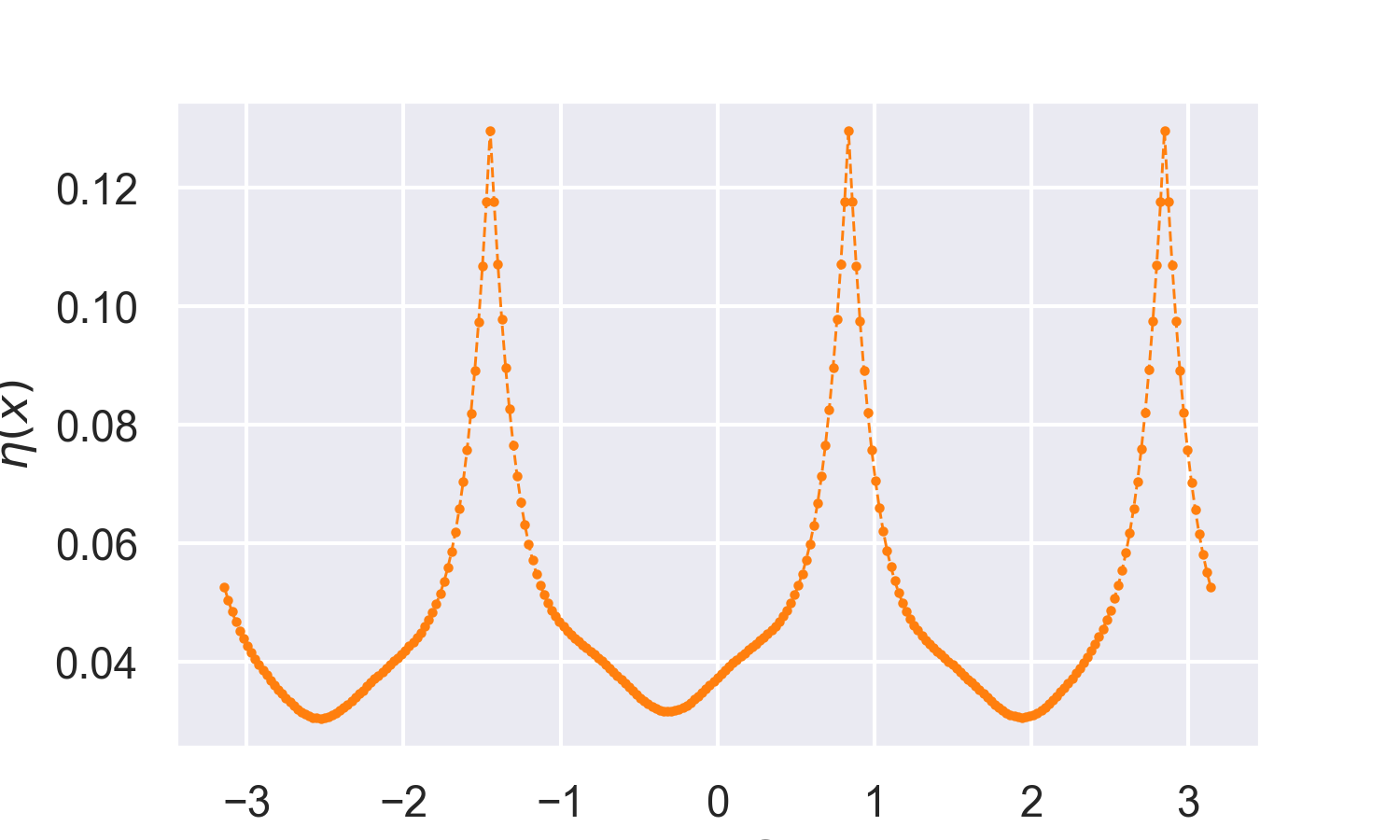}
  \caption{Real wages}
 \end{subfigure}
 \begin{subfigure}{0.5\columnwidth}
  \centering
  \includegraphics[width=\columnwidth]{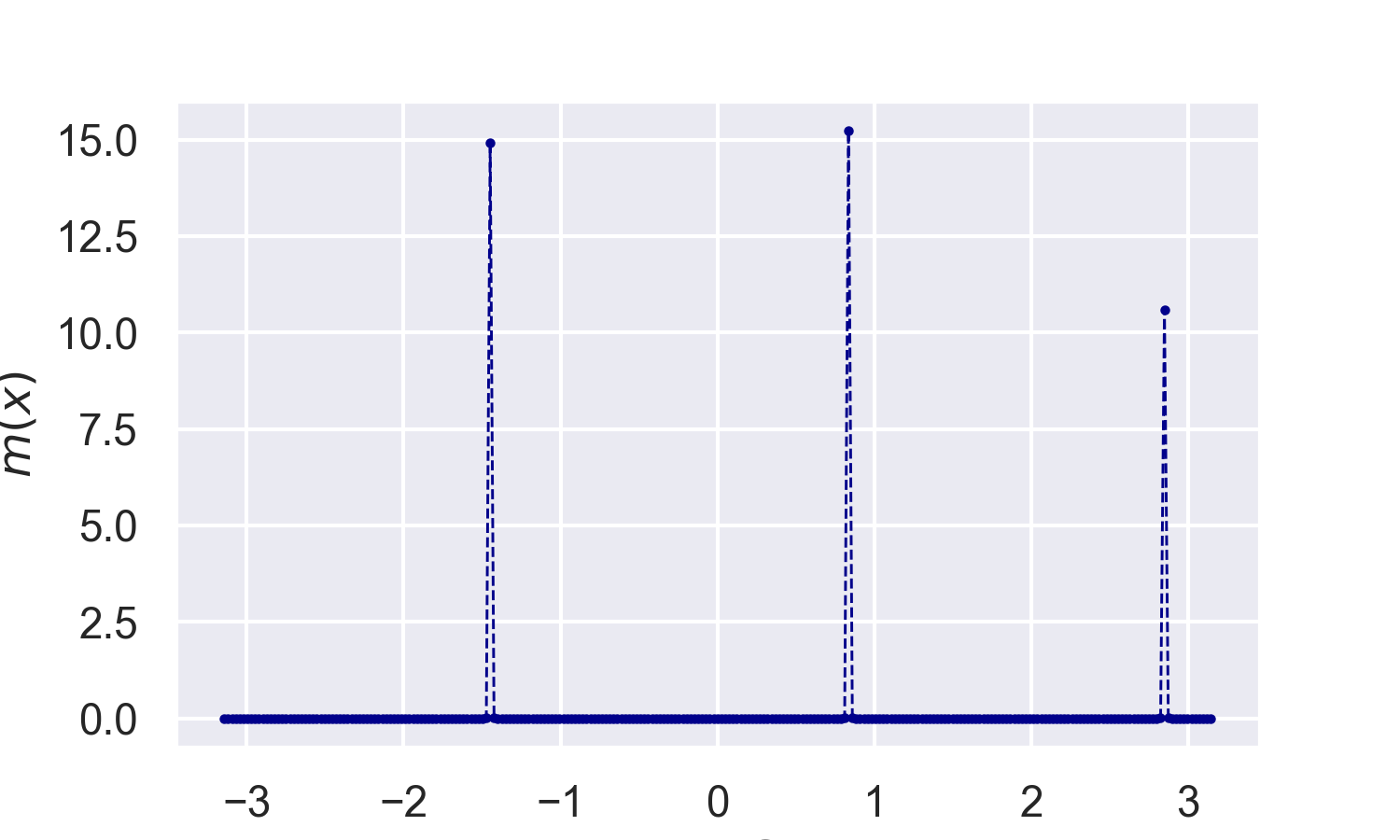}
  \caption{Manufacturing workers}
 \end{subfigure}
 \begin{subfigure}{0.5\columnwidth}
  \centering
  \includegraphics[width=\columnwidth]{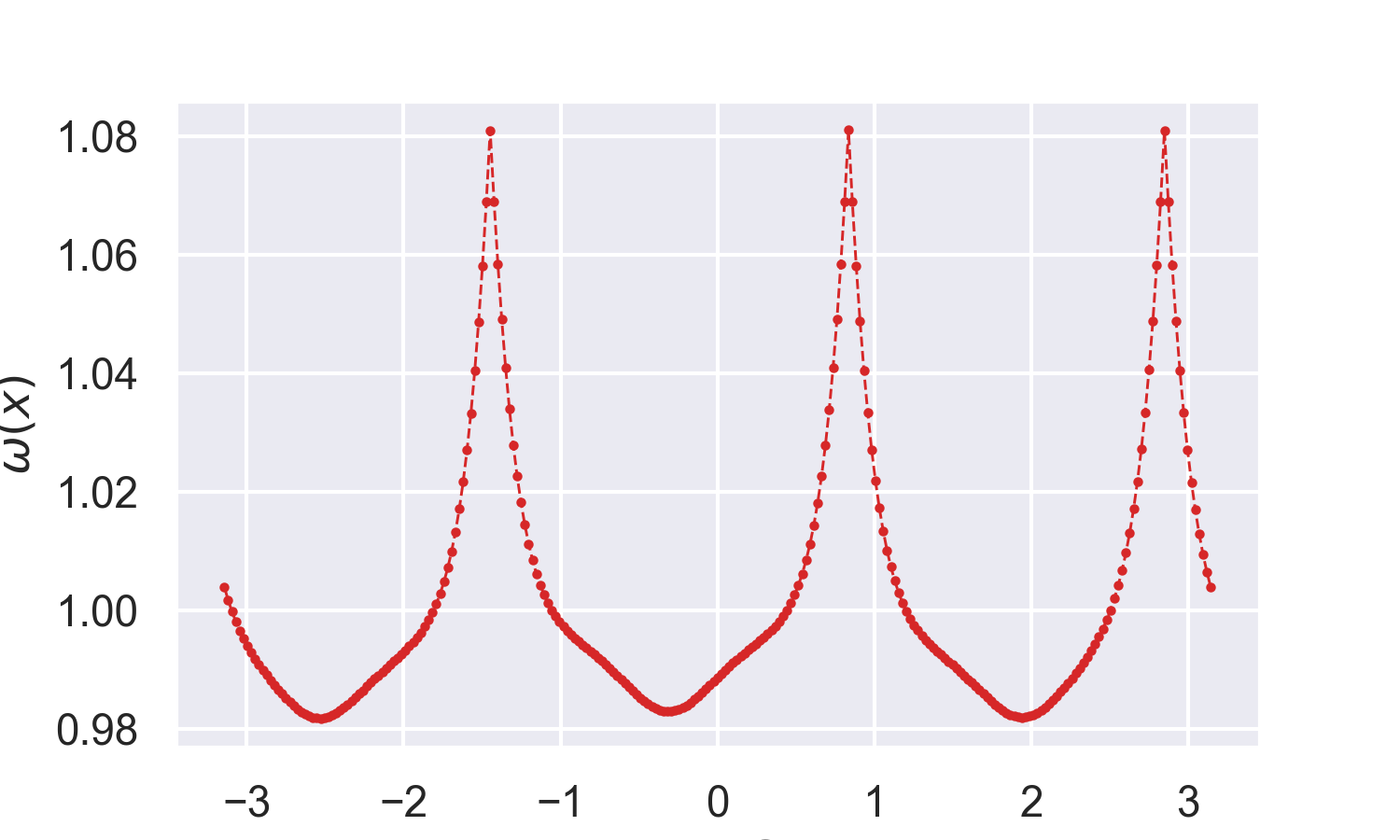}
  \caption{Real profits}
 \end{subfigure}
 \caption{A stationary solution for $\tau=1.1$ when other parameter values are $\mu=0.6$, $\sigma=5.0$, $v_n=1.0$, $v_m=1.0$, and $\rho=1.0$. Three cities are formed.}
 \label{fig:tau1p1}
\end{figure}

\vfill
\begin{figure}[H]
 \begin{subfigure}{0.5\columnwidth}
  \centering
  \includegraphics[width=\columnwidth]{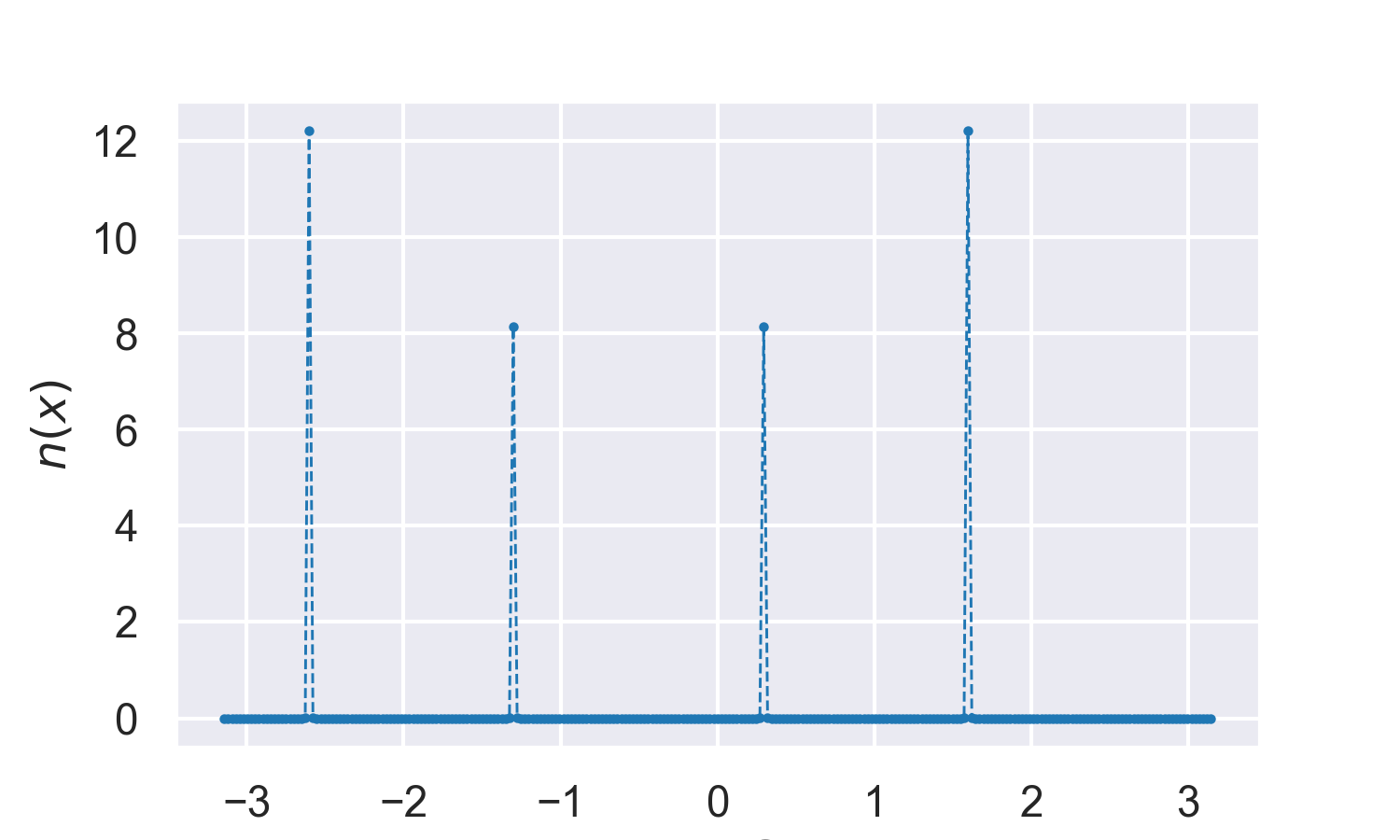}
  \caption{Firms}
 \end{subfigure}
 \begin{subfigure}{0.5\columnwidth}
  \centering
  \includegraphics[width=\columnwidth]{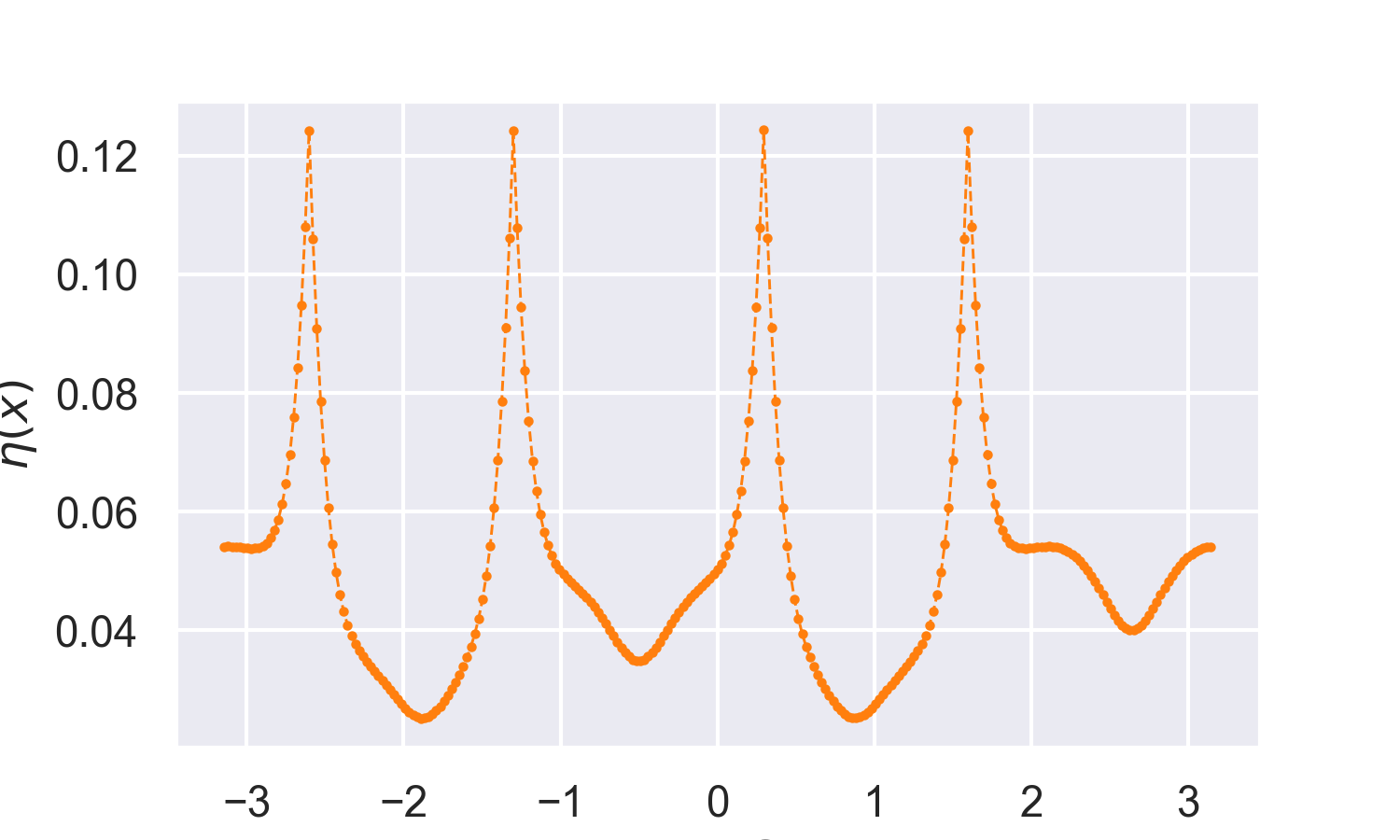}
  \caption{Real wages}
 \end{subfigure}
 \begin{subfigure}{0.5\columnwidth}
  \centering
  \includegraphics[width=\columnwidth]{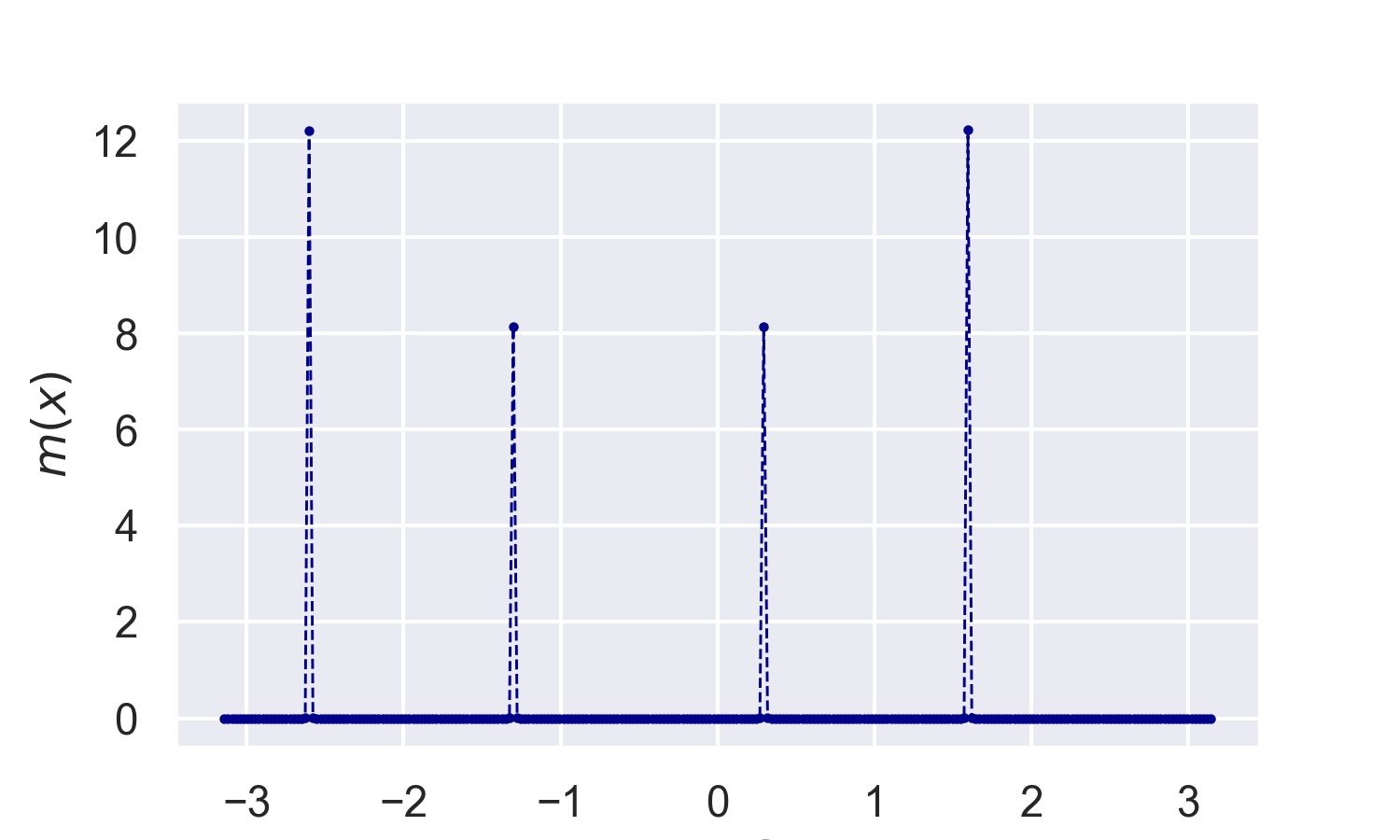}
  \caption{Manufacturing workers}
 \end{subfigure}
 \begin{subfigure}{0.5\columnwidth}
  \centering
  \includegraphics[width=\columnwidth]{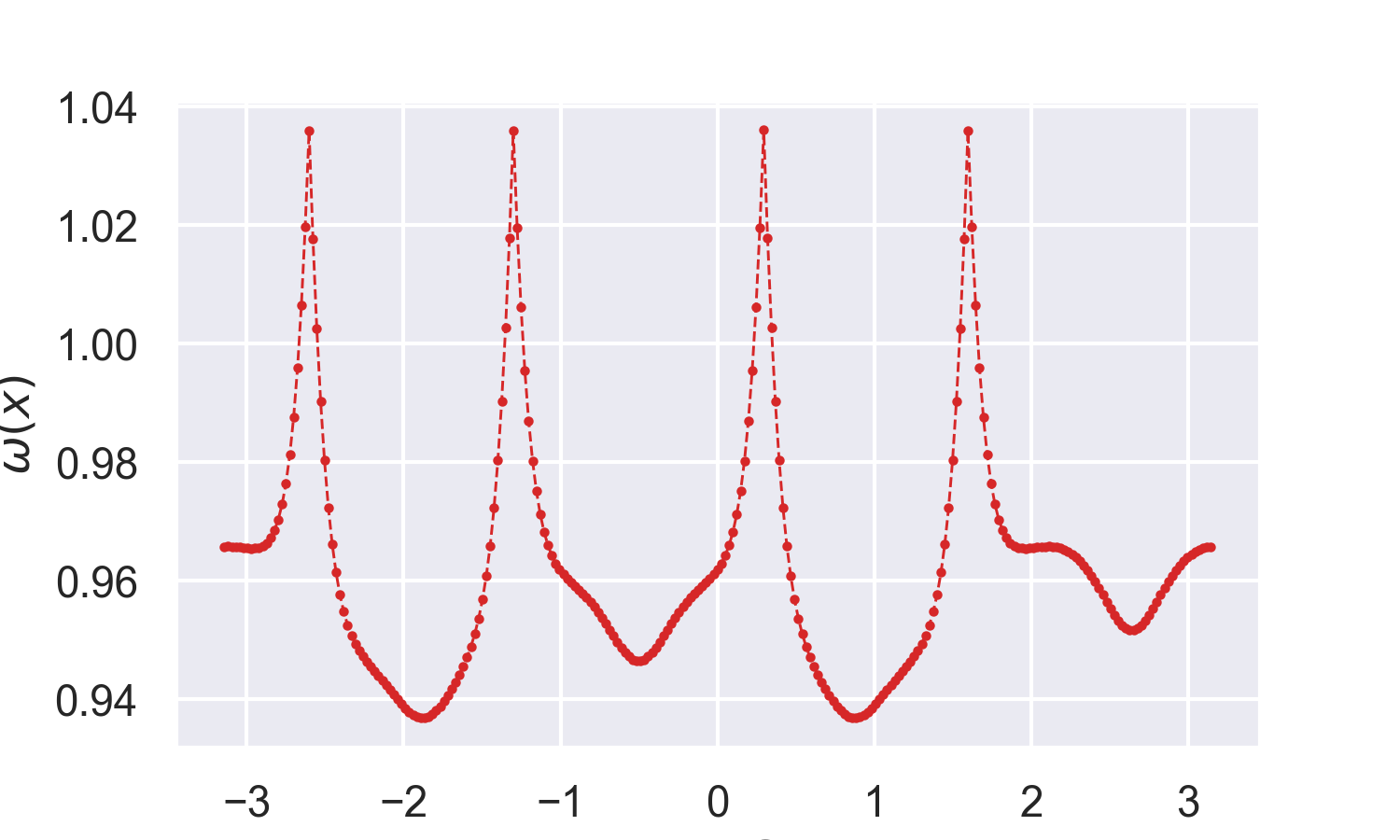}
  \caption{Real profits}
 \end{subfigure}
 \caption{A stationary solution for $\tau=1.7$ when other parameter values are $\mu=0.6$, $\sigma=5.0$, $v_n=1.0$, $v_m=1.0$, and $\rho=1.0$. Four cities are formed.}
 \label{fig:tau1p7}
\end{figure}
\vfill

\vfill
\begin{figure}[H]
 \begin{subfigure}{0.5\columnwidth}
  \centering
  \includegraphics[width=\columnwidth]{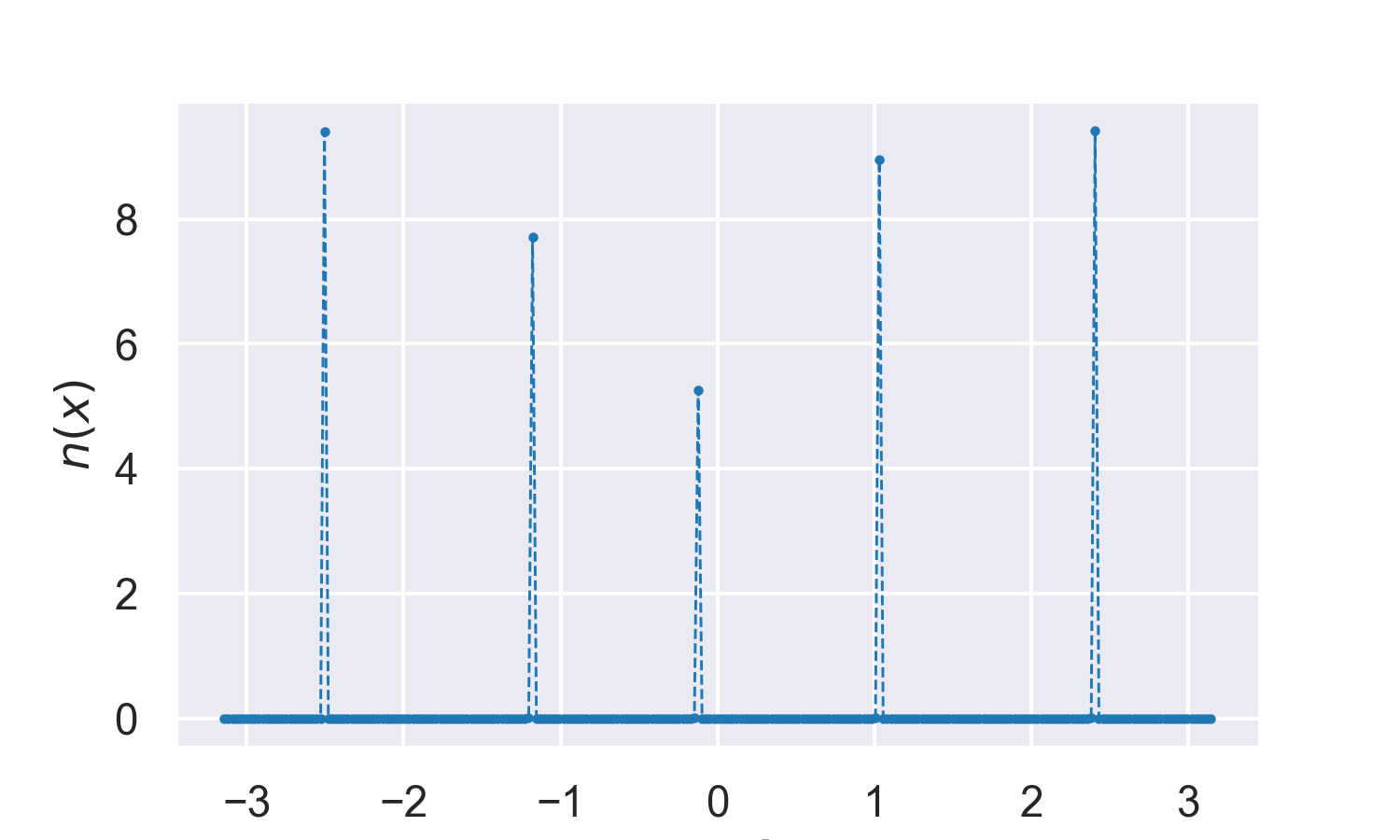}
  \caption{Firms}
 \end{subfigure}
 \begin{subfigure}{0.5\columnwidth}
  \centering
  \includegraphics[width=\columnwidth]{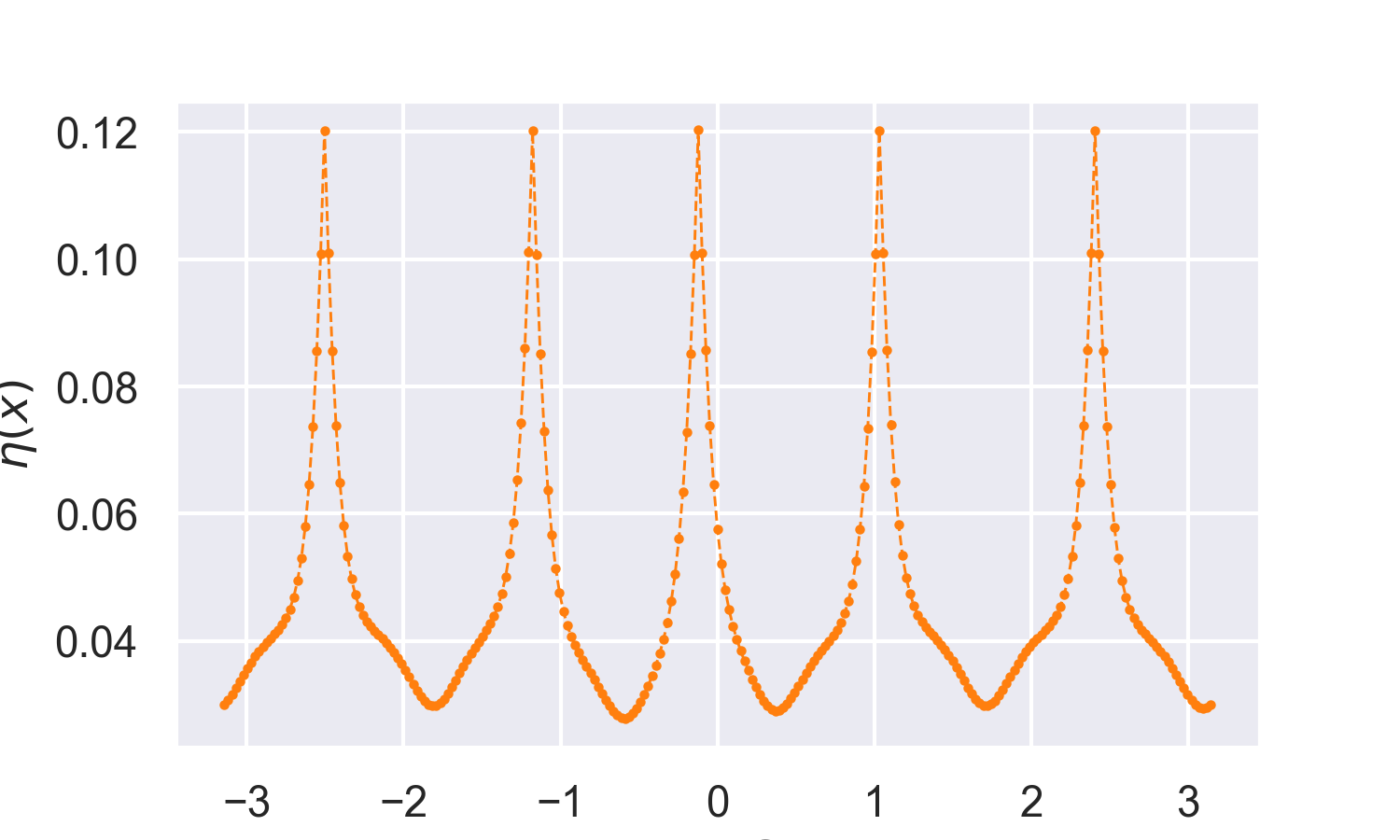}
  \caption{Real wages}
 \end{subfigure}
 \begin{subfigure}{0.5\columnwidth}
  \centering
  \includegraphics[width=\columnwidth]{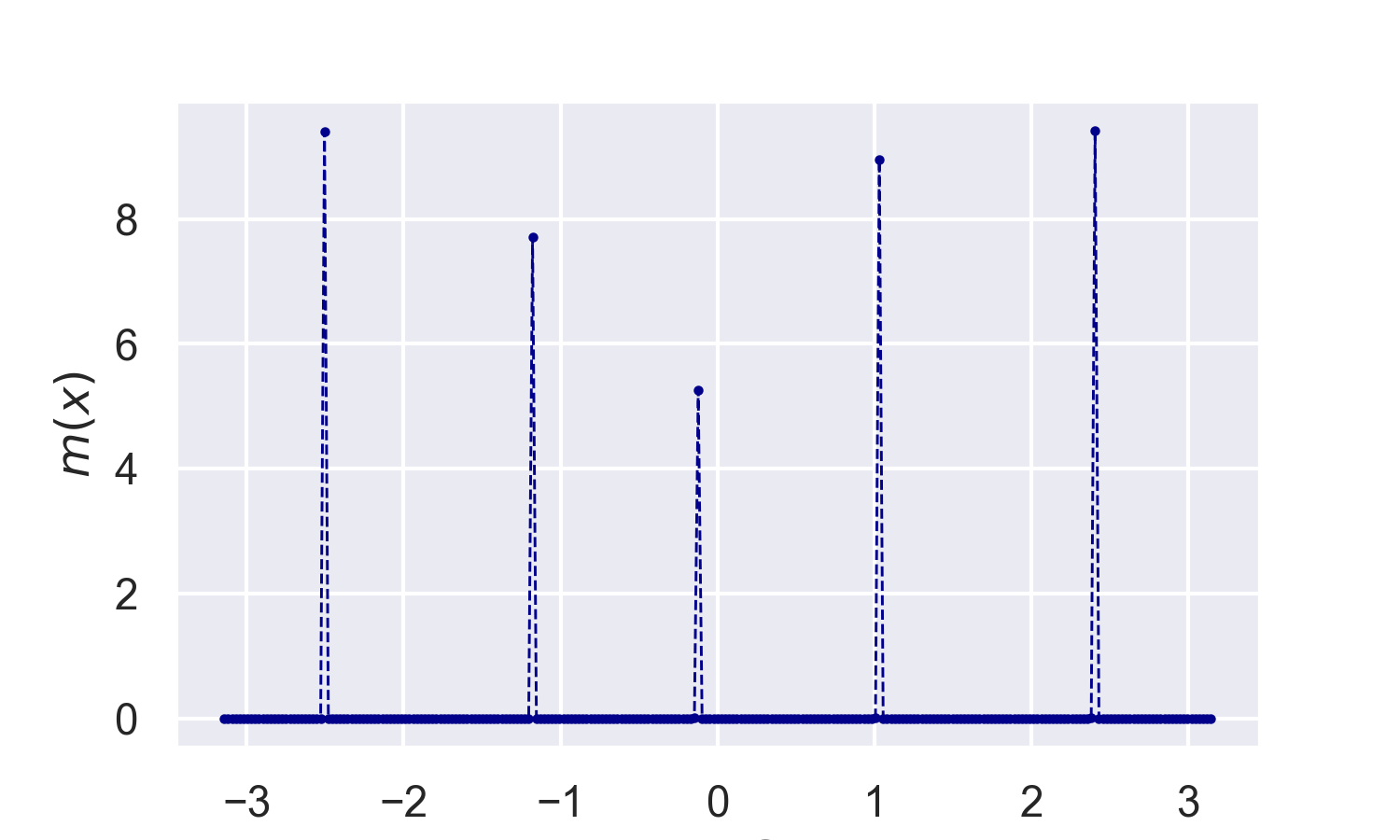}
  \caption{Manufacturing workers}
 \end{subfigure}
 \begin{subfigure}{0.5\columnwidth}
  \centering
  \includegraphics[width=\columnwidth]{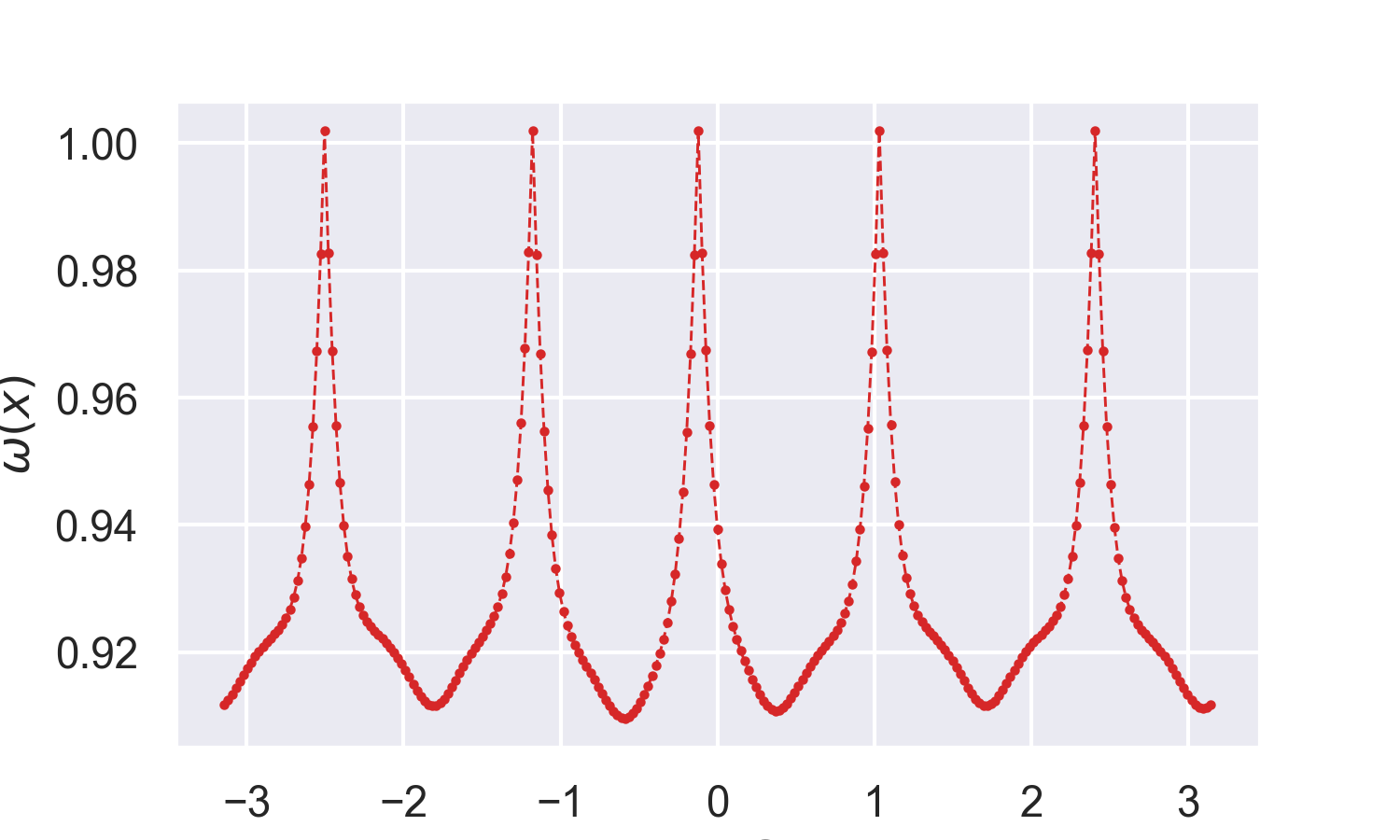}
  \caption{Real profits}
 \end{subfigure}
 \caption{A stationary solution for $\tau=2.0$ when other parameter values are $\mu=0.6$, $\sigma=5.0$, $v_n=1.0$, $v_m=1.0$, and $\rho=1.0$. Five cities are formed.}
 \label{fig:tau2p0}
\end{figure}
\vfill

\begin{figure}[H]
 \begin{subfigure}{0.5\columnwidth}
  \centering
  \includegraphics[width=\columnwidth]{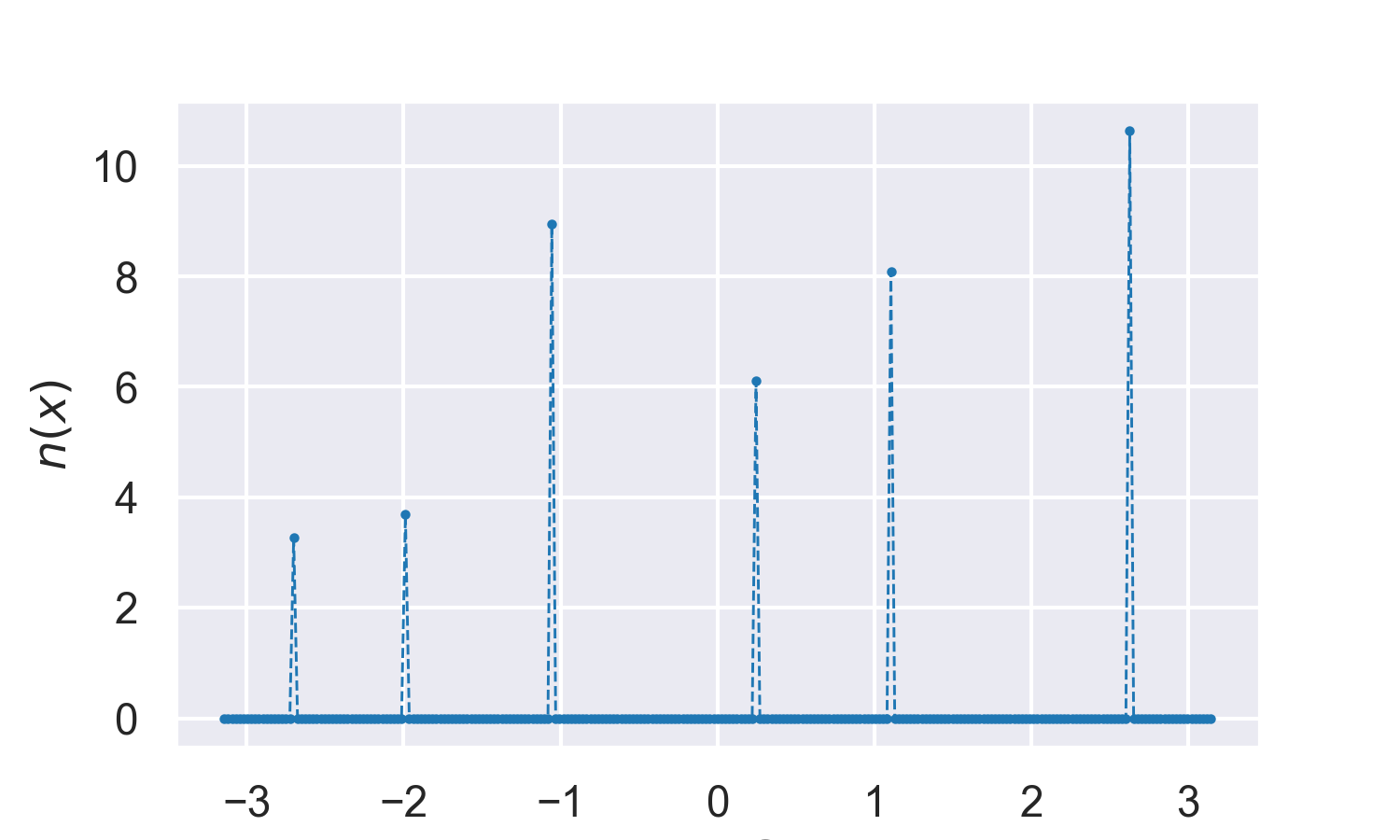}
  \caption{Firms}
 \end{subfigure}
 \begin{subfigure}{0.5\columnwidth}
  \centering
  \includegraphics[width=\columnwidth]{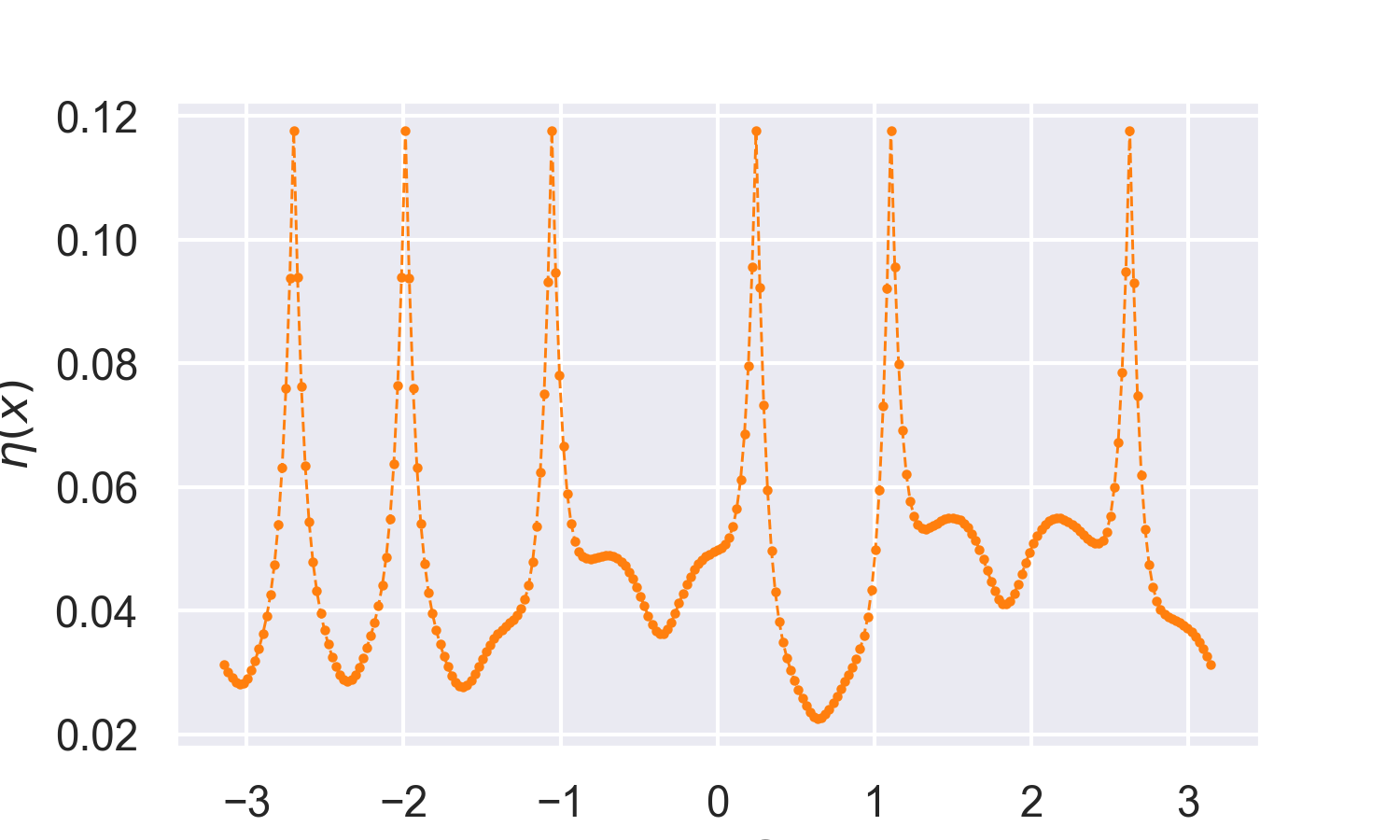}
  \caption{Real wages}
 \end{subfigure}
 \begin{subfigure}{0.5\columnwidth}
  \centering
  \includegraphics[width=\columnwidth]{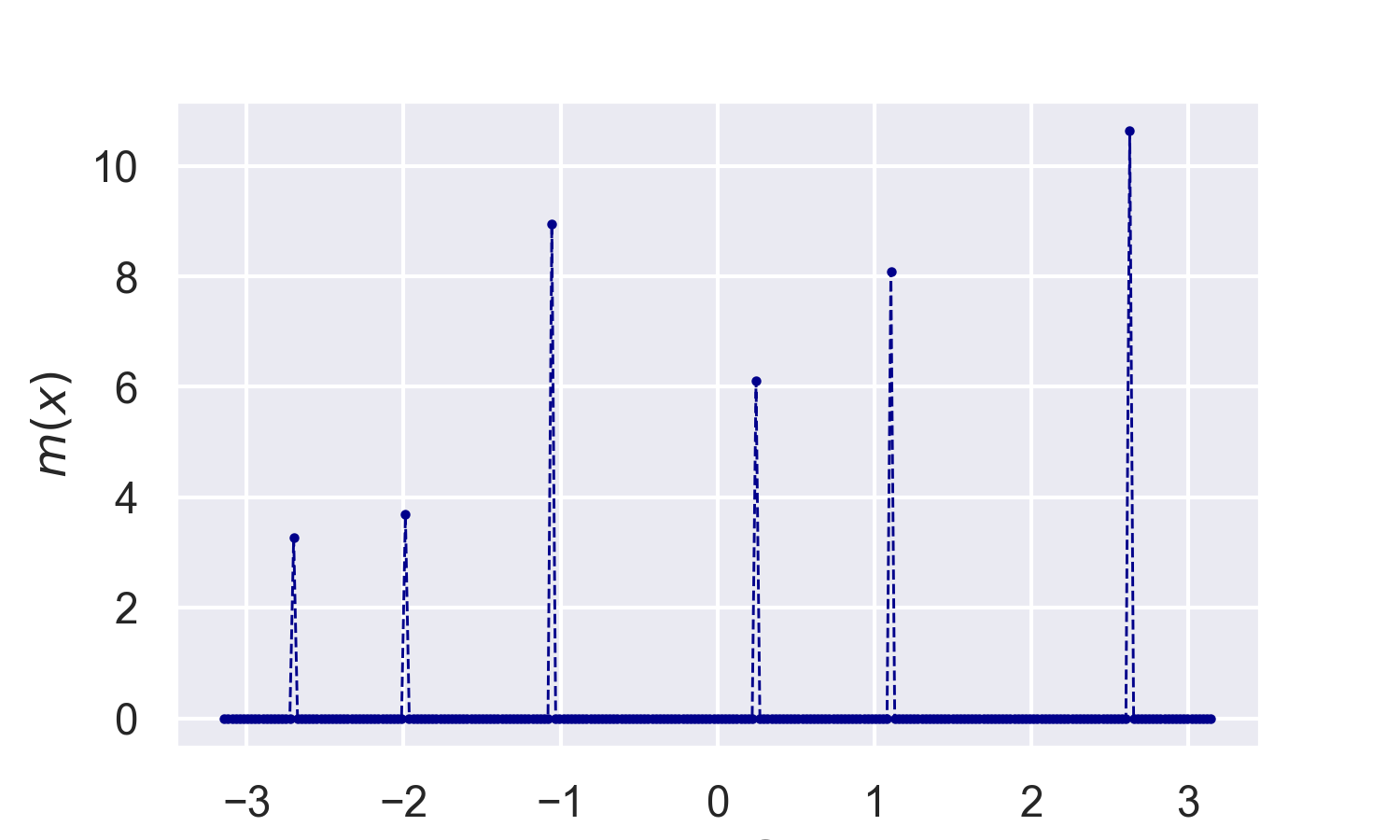}
  \caption{Manufacturing workers}
 \end{subfigure}
 \begin{subfigure}{0.5\columnwidth}
  \centering
  \includegraphics[width=\columnwidth]{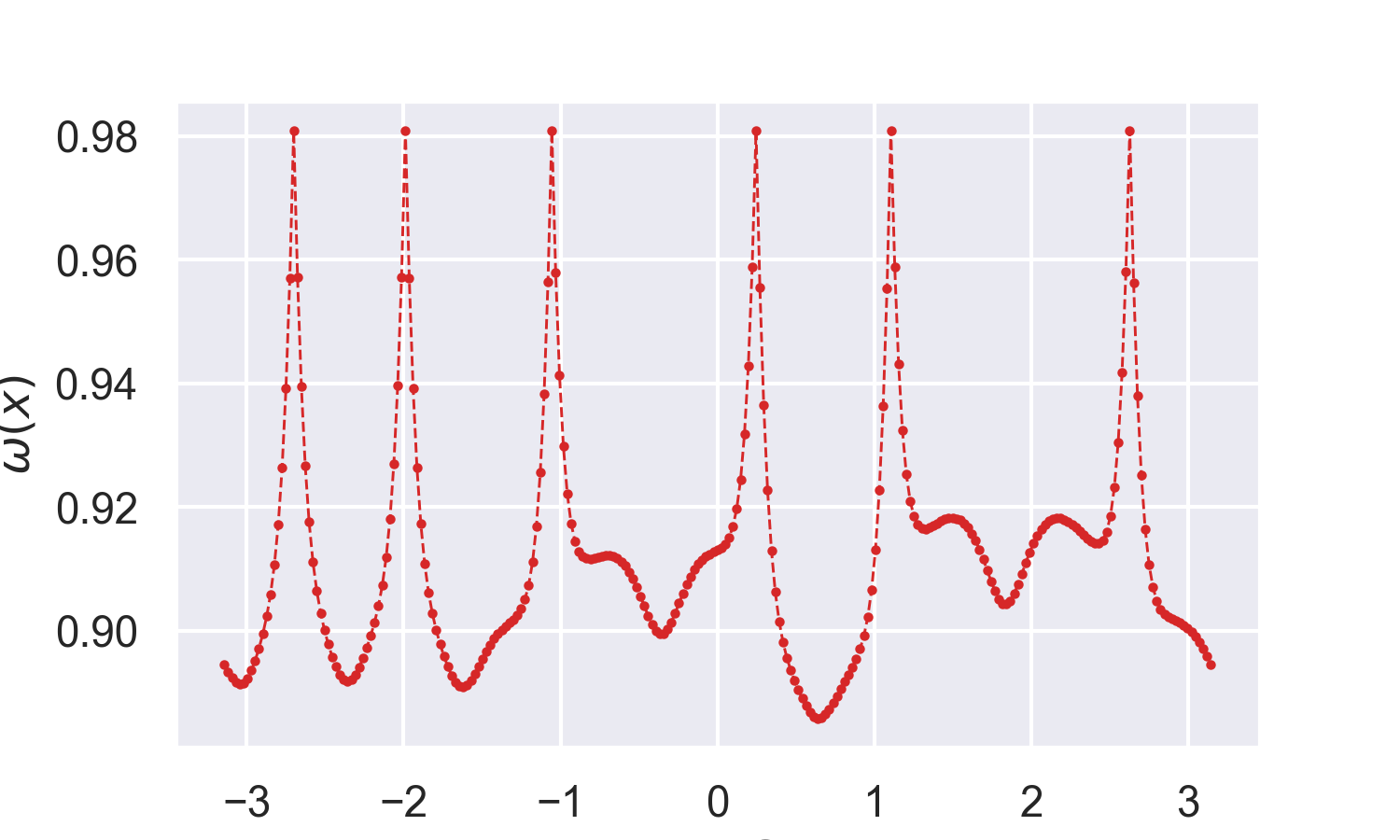}
  \caption{Real profits}
 \end{subfigure}
 \caption{A stationary solution for $\tau=2.6$ when other parameter values are $\mu=0.6$, $\sigma=5.0$, $v_n=1.0$, $v_m=1.0$, and $\rho=1.0$. Six cities are formed.}
 \label{fig:tau2p6}
\end{figure}

\newpage
\restoregeometry 
\noindent
{\bf Remark}: It is worth discussing the symmetry of the numerical stationary solutions. As shown in Figures \ref{fig:tau1p1}-\ref{fig:tau2p6}, when the number of cities is three or more, asymmetric equilibrium---meaning that the size of each city and the distances between adjacent cities are not uniform---is possible and is commonly observed.\footnote{Asymmetric equilibria are frequently observed in other NEG models as well (see, for example, \citet{Fab}, \citet{OhtakeYagi_Asym}, \citet{IkeOnTaka2019}, \citet{Ohtake2023cont}, and \citet{Ohtake2025agriculture})} City size and inter-city distances vary depending on initial values. There seems to be no robust relationship. It is difficult to grasp the economic and intuitive understanding of this matter. However, at the very least, it can be understood as demonstrating the historical path dependency of city size and distribution.

\vspace{1mm}
\noindent
{\bf Remark}: The model \eqref{0} contains fractional terms $n/m$ and $m/n$. Therefore, as solutions $n$ and $m$ asymptotically converge to spiky distributions, these fractional terms must approach the indeterminate form of $0/0$ at $x\notin {\rm supp}(\tilde{n})={\rm supp}(\tilde{m})$, or they may diverge to infinity before converging to stationary solutions. In most cases, computations do not break down, but in a few cases, an appropriate stationary solution may not be obtained. In those cases, the isomorphism observed in the figures above between the real profit distribution and the real wage distribution was broken.

\section{Conclusion}
We consider the model incorporating the dual migration process of firms and manufacturing workers. It has been found that the homogeneous stationary solution is unstable, and as transport costs fall, the eigenmodes with larger absolute spatial frequencies become destabilized first. Analytical discussions have shown that firms and manufacturing workers are distributed isomorphically in any stationary solution, whether homogeneous or non-homogeneous. Numerical simulations have shown that perturbed homogeneous stationary solutions eventually converge to non-homogeneous stationary solutions with spiky distributions of firms and manufacturing workers, whose spikes correspond to cities. The number of cities has been found to decrease as transport costs fall.

As a result, the solutions to the dual migration model proposed in this study exhibit behavior qualitatively similar to that of standard single-migration models. The question posed at the introduction ---``if both firms and workers could migrate independently, rather than one being subordinate to the other, what would the conclusion of NEG be?''--- has now been answered. The answer is that ``qualitatively, it does not change.'' In this regard, the present study provides strong evidence in support of the use of the single-migration model for NEG research.

\section{Appendix}

\subsection{Proof of Proposition \ref{th:Z01} \footnote{The proof is essentially same as that in \citet[pp.~495-496]{Ohtake2025agriculture}. }}\label{subsec:thZ01}

\subsubsection{When $k$ is even}

\vspace{1mm}
Putting $\alpha \rho=X$, we have
\[
Z_k 
=\frac{X^2}{k^2+X^2}.
\]
It is easy to see that
\[
\begin{aligned}
\lim_{X\to 0} \frac{X^2}{n^2+X^2} = 0,\\
\lim_{X\to \infty}\frac{X^2}{n^2+X^2} = 1.
\end{aligned}
\]
and that
\[
\frac{d}{dX}Z_k = \frac{2k^2X}{(k^2+X^2)^2} > 0.
\]

\subsubsection{When $k$ is odd}

\vspace{1mm}
Putting $X=\alpha \rho$, we have
\[
Z_k =
\frac{X^2(1+e^{-X\pi})}{(k^2+X^2)(1-e^{-X\pi})}.
\]
By L'Hopital's theorem, we have
\[
\lim_{X\to 0}Z_k =
\lim_{X\to 0}\frac{2X(1+e^{-\pi X})+X^2(-\pi e^{-\pi X})}{2X(1-e^{-\pi X})+(k^2+X^2)\pi e^{-\pi X}} 
= 0.
\]
It is clear that
\[
\lim_{X\to \infty}Z_k =1.
\]
Differentiating $Z_k$ by $X$ yields
\begin{equation}\label{ddXHbetanodd}
\frac{d}{dX}Z_k
= \frac{2k^2X(1-e^{-2\pi X})-2\pi X^2e^{-\pi X}(k^2+X^2)}{(k^2+X^2)^2(1-e^{-\pi X})^2}.
\end{equation}
To show that \eqref{ddXHbetanodd} is positive, we only have to prove that
\begin{equation}\label{epiXeminuspiXgt}
e^{\pi X}-e^{-\pi X} > \frac{\pi X}{k^2}(k^2+X^2).
\end{equation}
By expanding the left-hand side of \eqref{epiXeminuspiXgt} into the Maclaurin series, we obtain
\[
2\pi X +2\frac{(\pi X)^3}{3!}+2\frac{(\pi X)^5}{5!}+\cdots 
> \frac{\pi X}{k^2}(k^2+X^2).
\]
Therefore, subtracting the right-hand side from the left-hand side yields
\begin{equation}\label{piX2piX33!}
\pi X+\left\{2\frac{(\pi X)^3}{3!}-\frac{\pi}{k^2}X^3\right\}+2\frac{(\pi X)^5}{5!}+\cdots > 0
\end{equation}
Now, the content of the curly brackets
\[
2\frac{(\pi X)^3}{3!}-\frac{\pi}{k^2}X^3
\]
in \eqref{piX2piX33!} is positive. To see this, it is sufficient to check the case of $k=1$, and we see that
\[
\pi\left(\frac{\pi^2}{3}-1\right)X^3>0,
\]
which completes the proof.
\qed

\subsection{Numerical scheme}\label{subsec:numerical_scheme}
Let us see basic settings\footnote{This setting is the same as \citet{Ohtake2023cont}.} for the simulation in Section \ref{sec:simu}. The circle $S$ is identified with the interval $[-\pi,\pi)$. For a function $f$ on $[0,\infty)\times S$, we denote $f(t_h,x_i)=f^h_i$ and $f^h = [f^h_1, f^h_2, \cdots, f^h_I]$ where $x_i=-\pi+(i-1)\Delta x$ for $i=1,2,\cdots,256$ and $t_h=(h-1)\Delta t$ for $h=1,2,\cdots$. The step-size $\Delta x$ and $\Delta t$ are given as $\Delta x = \frac{2\pi}{256}$ and $\Delta t = \frac{1}{100}$, respectively. Approximating the integral $\int_{-\pi}^\pi f(t,x)dx$ by using the trapezoidal rule\footnote{The approximation of an integral over $S$ with the trapezoidal rule is equivalent to that with the Riemann sum as in \eqref{approxintegral}.} yields
\begin{equation}\label{approxintegral}
\int_{-\pi}^\pi f(t,x)dx \approx \sum_{i=1}^I f(t, x_i)\Delta x.
\end{equation}

Let us see an algorithm for obtaining numerical stationary solutions of \eqref{0}. Substituting the first and the third equations into the second one of \eqref{0}, we have 
\begin{equation}\label{fixedpointW}
W(t,x)=\frac{n(t,x)}{m(t,x)}\int_S \frac{(1-\mu)\ol{\phi}+\mu W(t,y)^{\frac{1}{\sigma}}m(t,y)}{\int_S n(t,z)W(t,z)^{\frac{1-\sigma}{\sigma}}T(y,z)^{1-\sigma}dz}T(x,y)^{1-\sigma}dy,
\end{equation}
where $W(t,x):=w(t,x)^{\sigma}$. To obtain the nominal wage for each $n(t,\cdot)$ and $m(t,\cdot)$, $t>0$, we have to solve the fixed point problem \eqref{fixedpointW}.\footnote{A formulation of the nominal wage equation of the original CP model as such a fixed point problem can be found in \citet{TabaEshiSakaTaka}.} Once a solution $W$ is found, the solutions $w$, $G$, $\eta$, and $\omega$ are obtained in a straightforward manner from $w(t,x)=W(t,x)^{\frac{1}{\sigma}}$ and the third to the fifth equations of \eqref{0}. Solving \eqref{fixedpointW} is theoretically difficult, but in practice, a numerical solution converges through repeated substitutions in most cases. That is, for each time step $h$, we only have to iterate the mapping $F:\mathbb{R}^I\to\mathbb{R}^I$ defined by
\[
[F({W^h})]_i := \frac{n_i^h}{m_i^h}\sum_{j=1}^I \frac{(1-\mu)\ol{\phi}_j+\mu {W_j^h}^{\frac{1}{\sigma}}m_j^h}{\sum_{k=1}^I n_k^h{W_k^h}^{\frac{1-\sigma}{\sigma}}T(y_j,z_k)^{1-\sigma}\Delta z}T(x_i,y_j)^{1-\sigma}\Delta y.
\]
where $i=1,2,\cdots,I$ until the maximal norm $\left\|F(W^h)-W^h\right\|_\infty < 10^{-10}$. The time evolution of the share functions $n$ and $m$ is computed by the explicit Euler method. We continue the simulation until $\left\|n^{h+1}-n^h\right\|_\infty<10^{-10}$ and $\left\|m^{h+1}-m^h\right\|_\infty<10^{-10}$.

\subsection{Proof of Proposition \ref{th: existence of critical point}}\label{subsec:existence of cp}
We begin by proving the following two lemmas.
\begin{lem}\label{lem:A/sigma-1<0}
Under the assumption of no black holes \eqref{nbh},
\[
\frac{A}{\sigma}-1 < 0
\]
for any $Z\in [0, 1]$
\end{lem}
\begin{proof}
Since $\sigma>1$ and $\delta(Z) > 0$, it is sufficient to prove that $\delta(Z)\left(A-\sigma\right)<0$. From \eqref{deltaZk} and \eqref{defA}, we have
\[
\begin{aligned}
\delta(Z)\left(A-\sigma\right)
&= \left(\sigma-2-\frac{\mu^2}{\sigma-1}\right)Z^2
+ \frac{\mu\sigma}{\sigma-1}Z -\sigma + 1 \\
&=:\varphi(Z).
\end{aligned}
\]
It is easy to see that
\[
\varphi(0) = -\sigma+1 < 0
\]
and
\begin{equation}\label{varphi1}
\begin{aligned}
\varphi(1) &= \frac{(1-\mu)\left(-\sigma+1+\mu\right)}{\sigma-1}\\
&< \frac{1-\mu}{\sigma-1}\left(-\frac{1}{1-\mu}+1+\mu\right) = \frac{-\mu^2}{\sigma-1} < 0
\end{aligned}
\end{equation}
In \eqref{varphi1}, the first inequality is due to the assumption of no black holes \eqref{nbh}. Hence, if the coefficient 
\begin{equation}\label{coeffZ2}
\sigma-2-\frac{\mu^2}{\sigma-1}
\end{equation}
of $Z^2$ in $\varphi(Z)$ is positive, then $\varphi(Z) < 0$ for all $Z\in[0,1]$.

Then, we next consider the case \eqref{coeffZ2} is negative, i.e., 
\begin{equation}\label{sig2sig2mu2}
(\sigma-2)(\sigma-1) - \mu^2 < 0.
\end{equation}
The quadratic function $\varphi$ has a vertex at 
\[
Z = Z_{\text{ver}} = \frac{-\mu\sigma}{2\left\{(\sigma-2)(\sigma-1)-\mu^2\right\}}>0
\]
and its value is
\[
\varphi(Z_{\text{ver}}) = \frac{1}{\sigma-1}\left[\frac{-\mu^2\sigma^2}{4\left\{(\sigma-1)(\sigma-2)-\mu^2\right\}} - (\sigma-1)^2\right] < 0
\]
Thus, $\varphi<0$ for all $Z\in[0, 1]$.
\end{proof}

\begin{lem}\label{lem:signA+BZZ*}
Let $Z^*$ be defined as
\begin{equation}\label{Z*}
Z^* := \frac{\mu(2\sigma-1)}{\sigma(1+\mu^2)-1}.
\end{equation}
Then,
\begin{equation}\label{signA+BZZ*}
\begin{aligned}
&A+B > 0 \hspace{3mm} \text{when} \hspace{3mm} Z < Z^*, \\
&A+B < 0 \hspace{3mm} \text{when} \hspace{3mm} Z > Z^*.
\end{aligned}
\end{equation}
\end{lem}
\begin{proof}
From \eqref{defA} and \eqref{defB}, we see that
\begin{equation}\label{AplusB}
A+B = 
\frac{1}{\delta(Z)}
\left[\left(-1-\frac{\mu^2\sigma}{\sigma-1}\right)Z^2 + \frac{\mu(2\sigma-1)}{\sigma-1}Z\right]
\end{equation}
Since $\delta(Z)>0$, the sign of \eqref{AplusB} is determined by the quadratic function in the brackets on the right-hand side. Then, it is easily verified that \eqref{signA+BZZ*} holds.
\end{proof}

\vspace{3mm}

The characteristic polynomial of the matrix \eqref{Mat} is
\begin{equation}\label{charapoly}
\lambda_i^2 - \left[v_n\mu\left(\frac{A}{\sigma} - 1\right) + v_mB\right]\lambda_i
-v_nv_m\mu(A+B) = 0
\end{equation}
where $\lambda_i$ ($i=1$, $2$) stands for any eigenvalues of \eqref{Mat}. By the relationships between roots and coefficients of polynomials known as Vieta's formula,\footnote{As a special case of Vieta's formula (see \citet[p.~275]{Kos1982IntAlg} for example), for the two roots $\lambda_1$ and $\lambda_2$ of a real-coefficient quadratic polynomial $a\lambda^2+b\lambda+c$, the sum of them satisfies $\lambda_1+\lambda_2=-b/a$.} we have
\begin{equation}\label{lambda1+lambda2}
\lambda_1 + \lambda_2 = v_n\mu\left(\frac{A}{\sigma}-1\right) + v_mB.
\end{equation}
Therefore, from \eqref{B<0} and Lemma \ref{lem:A/sigma-1<0}, we have 
\begin{equation}\label{lambda1+lambda2<0}
\lambda_1 + \lambda_2 < 0 \hspace{3mm} \forall Z \in(0, 1].
\end{equation}
Also by Vieta's formula,\footnote{As a special case of Vieta's formula, for the two roots $\lambda_1$ and $\lambda_2$ of a real-coefficient quadratic polynomial $a\lambda^2+b\lambda+c$, their product satisfies $\lambda_1\lambda_2=c/a$.} we have
\begin{equation}\label{lambda1lambda2}
\lambda_1\lambda_2 = -v_nv_m\mu(A+B).
\end{equation}
Therefore, from Lemma \ref{lem:signA+BZZ*} and \eqref{lambda1lambda2}, we have
\begin{equation}\label{ZZ*ZZ*lam1lam2sign<}
\lambda_1\lambda_2 < 0\hspace{3mm} \forall Z < Z^*
\end{equation}
and
\begin{equation}\label{ZZ*ZZ*lam1lam2sign>}
\lambda_1\lambda_2 > 0\hspace{3mm} \forall Z > Z^*.
\end{equation}
From \eqref{lambda1+lambda2<0}, it is verified that\footnote{Note that when $\lambda_1,\lambda_2 \in \mathbb{C}\setminus\mathbb{R}$, the two eigenvalues are conjugate.}
\[
\lambda_1,\lambda_2 \in \mathbb{C}\setminus\mathbb{R}\Rightarrow \Re \lambda_1 = \Re\lambda_2 < 0 \hspace{3mm}\forall Z \in(0, 1]
\]
and 
\[
\hspace{-20mm}\lambda_1,\lambda_2\in\mathbb{R}\Rightarrow \lambda_1<0\hspace{2mm}\text{or}\hspace{2mm}\lambda_2 < 0 \hspace{3mm}\forall Z \in(0, 1]
\]

Suppose that $Z<Z^*$. If $\lambda_1,\lambda_2 \in \mathbb{C}\setminus\mathbb{R}$, then $\lambda_1\lambda_2>0$ because they are conjugate. This contradicts to \eqref{ZZ*ZZ*lam1lam2sign<}. Thus, $\lambda_1, \lambda_2\in\mathbb{R}$ and one of the eigenvalues is positive, and the other is negative.

Suppose that $Z>Z^*$. If $\lambda_1,\lambda_2 \in \mathbb{C}\setminus\mathbb{R}$, then all the eigenvalues have the negative real parts from \eqref{lambda1+lambda2<0} because they are conjugate. If $\lambda_1,\lambda_2 \in \mathbb{R}$, both the eigenvalues are negative because of \eqref{lambda1+lambda2<0} and \eqref{ZZ*ZZ*lam1lam2sign>}.

From the above, it is found that $Z^*$ is the threshold value at which an eigenvalue with a positive real part appears if $Z_k<Z^*$, and the real parts of all eigenvalues are negative if $Z_k>Z^*$. Thus, the critical point $\tau^*_k$ is defined as $\tau$ such that 
\[
Z_k=Z^*.
\]

Finally, $Z^*<1$ is shown as follows. 
\begin{align}
& Z^* < 1 \nonumber\\
&\Leftrightarrow \mu(2\sigma-1)  < \sigma(1+\mu^2) - 1 \nonumber\\
&\Leftrightarrow \sigma(1+\mu^2-2\mu) -1+\mu > 0 \nonumber\\
&\Leftrightarrow (1-\mu)\left\{\sigma(1-\mu)-1\right\} > 0.\label{Z*<1last}
\end{align}
From the assumption of no black holes \eqref{nbh}, we have $\sigma(1-\mu)-1>0$. Thus, \eqref{Z*<1last} holds.
\qed

\subsection{Proof of Proposition \ref{th:criticalpoint_kk+2}}\label{subsec:criticalpoint_kk+2}

As shown in Lemma \ref{lem:signA+BZZ*}, the critical point $\tau_k^*$ is the value of $\tau$ that satisfies the equation
\begin{equation}\label{Z_k=Z^*}
Z_k = Z^*.
\end{equation}
Here, $Z_k$ and $Z^*$ are given by \eqref{Zk} and \eqref{Z*}, respectively. 
From the definition \eqref{Zk} of $Z_k$, the equation \eqref{Z_k=Z^*} for an even $|k|$ becomes
\begin{equation}\label{evenkZk=Z*}
\frac{\alpha^2\rho^2}{k^2+\alpha^2\rho^2} = Z^*.
\end{equation}
Recall that $\alpha$ is given by \eqref{gtc}. Then, solving the equation \eqref{evenkZk=Z*} gives 
\begin{equation}\label{monoinck}
\tau_k^* = \frac{|k|}{(\sigma-1)\rho}\sqrt{\frac{Z^*}{1-Z^*}},
\end{equation}
which clearly increases linearly in $|k|$ which is even.\footnote{For an even frequency $k$, we must consider the changes in frequencies: $k$, $k\pm 2$, $k\pm 4$, $\cdots$.} Recall that the $k$-th mode is unstable in the open interval $(0,\tau_k^*)$. Therefore, under any high transport costs, at least for cases where $|k|$ is even, there exists a mode that becomes unstable if $|k|$ is sufficiently large.\footnote{Even if a certain $k$-th mode is stabilized by making $\tau$ sufficiently large, it is possible to cause unstable modes to appear again by making $|k|$ sufficiently large within the range of even numbers.}
\qed

\subsection{Critical point $\tau_k^*$ is an increasing function of odd $|k|$}\label{tauk*explicit-oddk}

Our goal is to prove the following proposition.
\begin{prop}\label{th:criticalpoint_oddk}
Suppose that \eqref{nbh} holds. For $|k|=2l-1$ where $l\in\mathbb{N}$, the critical point $\tau_k^*$ is positive and satisfies that
\begin{align}
&\tau_{k+2}^* > \tau_k^*\hspace{2mm}\text{when\hspace{2mm}$k\geq 1$},\label{tauk*increasinginoddk_k>1}\\
&\tau_{k-2}^* > \tau_k^*\hspace{2mm}\text{when\hspace{2mm}$k\leq -1$.}\label{tauk*increasinginoddk_k<-1}
\end{align}
Furthermore,
\begin{equation}\label{tauk*toinfty}
\lim_{|k|\to\infty} \tau_k^* = \infty.
\end{equation}
\end{prop}
\begin{proof}
From \eqref{Zk}, the equation \eqref{Z_k=Z^*} which must be satisfied by the critical point $\tau=\tau_k^*$ for any odd $|k|$ is
\begin{equation}\label{ZkeqZ*odd}
\frac{\alpha^2\rho^2}{k^2+\alpha^2\rho^2}\cdot\frac{1+e^{-\alpha\rho\pi}}{1-e^{-\alpha\rho\pi}} = Z^*.
\end{equation}
By the hyperbolic function $\coth(s)=\frac{e^s+e^{-s}}{e^s-e^{-s}}$, the equation \eqref{ZkeqZ*odd} is rewritten as
\begin{equation}\label{ZkeqZ*odd2}
\frac{\alpha^2\rho^2}{k^2+\alpha^2\rho^2}\cdot\coth\left(\frac{\alpha\rho\pi}{2}\right) = Z^*.
\end{equation}
Although $k$ is an odd integer, we handle it as a real number in the following discussion. For any given $k\in\mathbb{R}$, let $\alpha=\alpha\left(|k|\right)$ be a solution to the equation \eqref{ZkeqZ*odd2}. Then 
\begin{equation}\label{def:alpha|k|=sig-1tauk*}
\alpha\left(|k|\right) = (\sigma-1)\tau_k^*
\end{equation}
by the definition \eqref{gtc} of the generalized transport coefficient. By the inverse function theorem, we see that
\begin{equation}\label{invfuncthm}
\frac{d\alpha}{d|k|}\left(|k|\right) = \frac{1}{\frac{d|k|}{d\alpha}(\alpha)}.
\end{equation}
Therefore, to show that $\frac{d\alpha}{d|k|}>0$, it suffices to show that $\frac{d|k|}{d\alpha}>0$.

Solving \eqref{ZkeqZ*odd2} for $|k|=|k|(\alpha)$ yields
\begin{equation}\label{kalpha}
|k|(\alpha) = \alpha\rho\sqrt{\frac{\coth\left(\frac{\alpha\rho\pi}{2}\right)-Z^*}{Z^*}}.
\end{equation}
Note that $\frac{d}{ds}\coth(s)=-\frac{1}{\sinh^2(s)}$. Then, differentiating \eqref{kalpha} by $\alpha$, we have
\begin{equation}\label{d|k|dalpha}
\frac{d|k|}{d\alpha}(\alpha) = \frac{4\rho\left[\coth\left(\frac{\alpha\rho\pi}{2}\right)-Z^*\right]\sinh^2\left(\frac{\alpha\rho\pi}{2}\right)-\alpha\rho^2\pi}{4\sqrt{Z^*}\sqrt{\coth\left(\frac{\alpha\rho\pi}{2}\right)-Z^*}\sinh^2\left(\frac{\alpha\rho\pi}{2}\right)}.
\end{equation}
Since the denominator of this expression is positive,\footnote{Recall that $Z^*\in(0,1)$ and that $\coth(s)>1$ for $s>0$.} we need only consider the numerator to show that $\frac{d|k|}{d\alpha}(\alpha)>0$.

Let the numerator of \eqref{d|k|dalpha} be $N(\alpha)$. Then,
\begin{equation}\label{Nalpha}
N(\alpha) = 4\rho\coth\left(\frac{\alpha\rho\pi}{2}\right)\sinh^2\left(\frac{\alpha\rho\pi}{2}\right)
- 4\rho Z^* \sinh^2\left(\frac{\alpha\rho\pi}{2}\right) -\alpha\rho^2\pi.
\end{equation}
In the right-hand side of \eqref{Nalpha}, applying
\begin{equation}\label{coth=cosh/sinh}
\coth(s) = \frac{\cosh(s)}{\sinh(s)}
\end{equation}
to the first term, and applying
\begin{equation}\label{2sinh^2=cosh2x-1}
2\sinh^2(s) = \cosh(2s) - 1
\end{equation}
to the second term, we have 
\begin{equation}\label{Nalpha_2nd}
N(\alpha) = 4\rho\cosh\left(\frac{\alpha\rho\pi}{2}\right)\sinh\left(\frac{\alpha\rho\pi}{2}\right)-2\rho Z^*\left[\cosh\left(\alpha\rho\pi\right) - 1\right] - \alpha\rho^2\pi.
\end{equation}
In the right-hand side of \eqref{Nalpha_2nd}, by applying
\begin{equation}\label{double-angle_sinh2x}
2\sinh(s) \cosh(s) = \sinh (2s)
\end{equation}
to the first term, we obtain
\begin{equation}\label{Nalpha_simple}
N(\alpha) = 2\rho\sinh\left(\alpha\rho\pi\right) -2\rho Z^*\left[\cosh\left(\alpha\rho\pi\right)-1\right] -\alpha\rho^2\pi.
\end{equation}
It is easy to see that
\begin{equation}\label{N0=0}
N(0) = 0.
\end{equation}

Let us consider the derivative of $N(\alpha)$. Since $\frac{d}{ds}\sinh(s)=\cosh(s)$ and $\frac{d}{ds}\cosh(s)=\sinh(s)$, differentiating \eqref{Nalpha_simple} yilelds
\begin{equation}\label{ddalphaN}
\frac{d}{d\alpha}N(\alpha) = 2\rho^2\pi\cosh\left(\alpha\rho\pi\right)-2\rho^2\pi Z^*\sinh\left(\alpha\rho\pi\right)-\rho^2\pi.
\end{equation}
It is easy to see that
\begin{equation}\label{ddalphaN0>0}
\frac{d}{d\alpha}N(0) = \rho^2\pi > 0.
\end{equation}
We only need to consider the range where $|k|^2\geq 1$. Solving the inequality of $|k|^2\geq 1$ obtained from \eqref{kalpha} for $Z^*$ gives
\begin{equation}\label{rangeofZ*fork>=1}
Z^* \leq \frac{\alpha^2\rho^2\coth\left(\frac{\alpha\rho\pi}{2}\right)}{1+\alpha^2\rho^2}.
\end{equation}
From \eqref{ddalphaN} and \eqref{rangeofZ*fork>=1}, we see that
\begin{equation}\label{ddalphaNineq}
\frac{d}{d\alpha}N(\alpha) \geq 
2\rho^2\pi\cosh\left(\alpha\rho\pi\right)-\frac{2\alpha^2\rho^4\pi}{1+\alpha^2\rho^2} \coth\left(\frac{\alpha\rho\pi}{2}\right)\sinh\left(\alpha\rho\pi\right)-\rho^2\pi.
\end{equation}
In the second term of the right-hand side of \eqref{ddalphaNineq}, using \eqref{coth=cosh/sinh} and \eqref{double-angle_sinh2x} gives
\begin{equation}\label{ddalphaNineq_2nd}
\frac{d}{d\alpha}N(\alpha) \geq 
2\rho^2\pi\cosh\left(\alpha\rho\pi\right)-\frac{4\alpha^2\rho^4\pi}{1+\alpha^2\rho^2}\cosh^2\left(\frac{\alpha\rho\pi}{2}\right)-\rho^2\pi.
\end{equation}
In the first term of the right-hand side of \eqref{ddalphaNineq_2nd}, by using
\begin{equation}\label{cosh2x=cosh^2x+sinh^2x}
\cosh(2s) = \cosh^2(s) + \sinh^2(s),
\end{equation}
we finally obtain the inequality
\begin{equation}\label{ddalphaNfinal}
\frac{d}{d\alpha}N(\alpha) \geq
2\rho^2\pi\left[\frac{1-\alpha^2\rho^2}{1+\alpha^2\rho^2}\cosh^2\left(\frac{\alpha\rho\pi}{2}\right) + \sinh^2\left(\frac{\alpha\rho\pi}{2}\right) - \frac{1}{2}\right].
\end{equation}
We only have to show that the right-hand side of \eqref{ddalphaNfinal} is positive. That is equivalent to
\begin{equation}\label{ddalphaNineqRHSenum}
(1-\alpha^2\rho^2)\cosh^2\left(\frac{\alpha\rho\pi}{2}\right)+(1+\alpha^2\rho^2)\sinh^2\left(\frac{\alpha\rho\pi}{2}\right)-\frac{1}{2}(1+\alpha^2\rho^2) > 0
\end{equation}
By using \eqref{cosh2x=cosh^2x+sinh^2x} and $\cosh^2(s) - \sinh^2(s) = 1$, we have 
\begin{align}
\eqref{ddalphaNineqRHSenum} &\iff \cosh(\alpha\rho\pi) - \frac{3}{2}\alpha^2\rho^2 - \frac{1}{2} > 0\\
&\iff \left(1+\frac{\alpha^2\rho^2\pi^2}{2}+\frac{\alpha^4\rho^4\pi^4}{4!}+\cdots\right) - \frac{3}{2}\alpha^2\rho^2 - \frac{1}{2} > 0\\
&\iff \frac{1}{2} + \frac{\alpha^2\rho^2}{2}\left(\pi^2-3\right)+\frac{\alpha^4\rho^4\pi^4}{4!}+\cdots > 0.
\end{align}
Here, we use the Maclaurin series for $\cosh(\alpha\rho\pi)$. The last inequality is correct. This, together with \eqref{ddalphaNfinal} and  \eqref{ddalphaNineqRHSenum}, leads to
\begin{equation}\label{ddalphaN>0}
\frac{d}{d\alpha}N(\alpha) > 0,\quad\forall \alpha>0.
\end{equation}
From \eqref{N0=0}, \eqref{ddalphaN0>0}, and \eqref{ddalphaN>0}, we obtain $N(\alpha) > 0,~\forall \alpha>0$. Thus, we finally obtain
\begin{equation}\label{dkdalpha>0}
\frac{d|k|}{d\alpha}(\alpha) > 0
\end{equation}
when $|k|\geq 1$.

We can now show that $\alpha\left(|k|+2\right)>\alpha\left(|k|\right)$. From \eqref{invfuncthm} and \eqref{dkdalpha>0}, we have
\begin{equation}
\frac{d\alpha}{d|k|}\left(|k|\right)>0,\quad\forall|k|\geq 1.
\end{equation}
Therefore,
\[
\alpha\left(|k|+2\right)-\alpha\left(|k|\right) = \int_{|k|}^{|k|+2}\frac{d\alpha}{d|k|}\left(|k|\right)d|k| > 0.
\]
This immediately yields \eqref{tauk*increasinginoddk_k>1} and \eqref{tauk*increasinginoddk_k<-1} from the definition \eqref{def:alpha|k|=sig-1tauk*} of $\alpha\left(|k|\right)$. 

We can also show that $\lim_{|k|\to\infty}\alpha\left(|k|\right)=\infty$. Since $\coth(s)\to 1$ as $s\to \infty$, it is observed from \eqref{kalpha} that
\[
|k|(\alpha)\to\alpha\rho\sqrt{\frac{1-Z^*}{Z^*}}
\]
as $\alpha\to\infty$. See Fig.~\ref{fig:k=k(alpha)}, which illustrates a sketch of the graph of $|k|(\alpha)$.\footnote{This is a rough sketch and is not based on precise values.} Hence, 
\[
\alpha\left(|k|\right) \to \frac{|k|}{\rho}\sqrt{\frac{Z^*}{1-Z^*}}
\]
as $|k|\to\infty$. This implies that $\lim_{|k|\to\infty}\alpha\left(|k|\right)=\infty$, which immediately yields \eqref{tauk*toinfty} from the definition \eqref{def:alpha|k|=sig-1tauk*} of $\alpha\left(|k|\right)$.
\begin{figure}[H]
\begin{center}
\begin{tikzpicture}[scale=0.6]
\def\Z{0.7}
\begin{axis}[
    axis lines=middle,
    xmin=0, xmax=11,
    ymin=0, ymax=7,
    xlabel=$\alpha$,
    ylabel=$|k|$,
    every axis x label/.style={at={(ticklabel* cs:1.0)}, anchor=west},
    every axis y label/.style={at={(ticklabel* cs:1.0)}, anchor=south},
    grid=none,
    xtick=\empty,
    ytick=\empty,
    extra x ticks={0},
    extra x tick labels={O},
    extra x tick style={ticklabel style={anchor=north east}},
    extra y ticks={1},
    extra y tick labels={$1$},
    extra y tick style={
        grid=major,
        grid style={dashed, gray},
        ticklabel style={anchor=east}
    },
    samples=100,
    domain=0.01:10
]
\addplot[thick, blue, smooth] {
    x * sqrt( ( (exp(pi*x/2) + exp(-pi*x/2)) / (exp(pi*x/2) - exp(-pi*x/2)) - \Z ) / \Z )
};
\end{axis}
\end{tikzpicture}
\caption{Sketch of the graph of $|k|(\alpha)$}\label{fig:k=k(alpha)}
\end{center}
\end{figure}

\end{proof}

\subsection{Proof of Proposition \ref{prop:neqm}}\label{proof:neqm}
From \eqref{maximizedprofit}, we have
\begin{equation}\label{maximizedprofitrewritten}
(\sigma-1)\frac{g(x)}{w(x)} = \frac{Mm(x)}{Nn(x)}
\end{equation}
Applying \eqref{rw} and \eqref{rp} to \eqref{maximizedprofitrewritten}, we have 
\begin{equation}\label{fracetaomegafrac1sigm1fracMmNn}
\frac{\eta(x)}{\omega(x)} = \frac{1}{\sigma-1}\frac{Mm(x)}{Nn(x)}.
\end{equation}
From the replicator equations \eqref{dynn} and \eqref{dynm}, since $\frac{\partial \tilde{n}}{\partial t}=\frac{\partial \tilde{m}}{\partial t}=0$ in stationary solutions, we have
\[
\tilde{\eta}(x)=\tilde{\eta}(y),\hspace{3mm}\forall x,~y \in {\rm supp}(\tilde{n})
\]
and
\[
\tilde{\omega}(x)=\tilde{\omega}(y),\hspace{3mm}\forall x,~y \in {\rm supp}(\tilde{m}).
\]
These immediately yield 
\begin{equation}\label{fracetaomegaequals}
\frac{\tilde{\eta}(x)}{\tilde{\omega}(x)} = \frac{\tilde{\eta}(y)}{\tilde{\omega}(y)},\hspace{3mm}\forall x,~y\in {\rm supp}(\tilde{n})={\rm supp}(\tilde{m}).
\end{equation}
Applying \eqref{fracetaomegafrac1sigm1fracMmNn} to \eqref{fracetaomegaequals} gives
\[
\frac{\tilde{m}(x)}{\tilde{n}(x)} = \frac{\tilde{m}(y)}{\tilde{n}(y)},\hspace{3mm}\forall x,~y\in {\rm supp}(\tilde{n})={\rm supp}(\tilde{m})
\]
which implies that
\begin{equation}\label{mpron}
\tilde{m}(x) = K\tilde{n}(x),\hspace{3mm} \forall x\in  {\rm supp}(\tilde{n})={\rm supp}(\tilde{m}).
\end{equation}
for a certain constant $K:=\frac{\tilde{m}(y)}{\tilde{n}(y)}>0$. By integrating both sides of  \eqref{mpron} on $S$, and using \eqref{conservation}, we obtain $K\equiv 1$. Thus, \eqref{mpron} imediately gives
\[
\tilde{n}(x)=\tilde{m}(x),\hspace{3mm} \forall x\in  {\rm supp}(\tilde{n})={\rm supp}(\tilde{m}).
\]
\qed

\vspace{10mm}
\noindent
{\large {\bf Statements and Declarations}}

\vspace{3mm}
\noindent
{\small {\bf Competing Interests}:\\
The author has no competing interests to declare.}

\vspace{3mm}
\noindent
{\small{\bf Funding}:\\
No funding was received to conduct this study.}

\vspace{10mm}
\noindent
{\bf Acknowledgements}

\vspace{3mm}
The author would like to express sincere gratitude to the anonymous reviewers. Their insightful and constructive comments have improved the paper. In particular, the analytical result in Subsection \ref{subsec:analytical} was proposed by one of the reviewers. The naming of ``co-location'' and ``proportionality'' was also given by the reviewer. These results have greatly enhanced the generality of the main claims of this paper.

\bibliographystyle{econ-aea}

\newpage
\pagenumbering{roman}
\listoffigures

\end{document}